\documentclass[11pt]{article}
\usepackage{graphicx,amsthm,amsmath,amssymb, subfiles} %

\usepackage[normalem]{ulem}
\usepackage{enumitem}
\usepackage{amsthm}
\usepackage{thmtools}
\usepackage[margin=2.5cm]{geometry}
\usepackage{hyperref}
\usepackage{cleveref}
\usepackage{tikz}
\usetikzlibrary{
  matrix,
  calc
}
\usepackage[svgnames]{xcolor}
\usepackage{xspace}
\usepackage{mathrsfs}
\usepackage[draft]{commenting}

\usepackage{tikzit}

\definecolor{darkgreen}{RGB}{1,50,32}

\declareauthor{n}{Navid}{green!40}
\declareauthor{s}{Suhail}{violet}
\declareauthor{f}{Francesca}{red}
\declareauthor{b}{Bruno}{blue}
\declareauthor{cite}{CITE}{red!30}
\declareauthor{errbefore}{before}{red!30}
\declareauthor{errafter}{after}{DarkGreen}

\newcommand{\AC}{{\mathsf {AC}}}
\newcommand{\TC}{{\mathsf {TC}}}
\newcommand{\NC}{\mathsf{NC}}
\newcommand{\ZO}{\{0,1\}}
\newcommand{\eps}{\varepsilon}
\newcommand{\cF}{{\mathcal F}}
\newcommand{\F}{ \mathbb{F}_{2^n}}
\newcommand{\N}{ \mathbb{N}}
\newcommand{\sfD}{{\mathsf D}}

\newcommand{\sfE}{{\mathsf E}}

\newcommand{\Bal}{{\mathsf{Bal}}}
\newcommand{\Unif}{\mathsf{Unif}}

\newcommand{\poly}{\mathsf{poly}}

\newcommand{\E}{\mathop{\mathbb E}}

\newtheorem{theorem}{Theorem}
\newtheorem{lemma}[theorem]{Lemma}
\newtheorem{conjecture}[theorem]{Conjecture}
\newtheorem{corollary}[theorem]{Corollary}
\newtheorem{claim}[theorem]{Claim}

\newtheorem{question}[theorem]{Question}
\theoremstyle{definition}
\newtheorem{definition}{Definition}[section]
\newtheorem{remark}{Remark}

\DeclareMathOperator{\calC}{\mathcal{C}}
\DeclareMathOperator{\Q}{\{0, 1\}}
\DeclareMathOperator{\mcsp}{\mathsf{MCSP}}

\title{The Switching Lemma shows\\
    what the Switching Lemma cannot prove:\\
    an unconditional natural-proofs barrier %
    }
\author{Bruno Loff\thanks{LASIGE, Faculdade de Ci{\^e}ncias, Universidade de Lisboa, Portugal. Funded by the European Union (ERC, HOFGA, 101041696). Views and opinions expressed are however those of the author(s) only and do not necessarily reflect those of the European Union or the European Research Council. Neither the European Union nor the granting authority can be held responsible for them.
Also supported by FCT through the LASIGE Research Unit, ref. UID/00408/2025. Emails: \texttt{bruno.loff@gmail.com, suhail.sherif@gmail.com, fr.ugazio@gmail.com}}\\
    \and Suhail Sherif\footnotemark[1]
    \and Navid Talebanfard\thanks{University of Sheffield, Sheffield, United Kingdom. Supported by Royal Society International Exchanges Grant IES\textbackslash R3\textbackslash 243239. Email: \texttt{n.talebanfard@sheffield.ac.uk}}
    \and Francesca Ugazio\footnotemark[1]}
\date{}

\newcommand{\PTIME}{\ensuremath{\mathsf{P}}\xspace}
\newcommand{\NP}{\ensuremath{\mathsf{NP}}\xspace}

\begin{document}

\maketitle

\begin{abstract}

The \emph{natural-proofs barrier} of Razborov and Rudich (JCSS'97) begins with the striking observation that all known lower-bound proofs follow a certain pattern: when showing that a function $F$ is hard, along the way the proof provides us with a \textit{distinguisher}, namely, an efficient algorithm which,  when given a truth-table $f$ as input, can distinguish the case when $f$ is an easy function from the case when $f$ is a random function. They called such lower-bound proofs \textit{natural proofs}. They then showed that under standard cryptographic assumptions, natural proofs cannot show, e.g., superpolynomial lower-bounds against Boolean circuits.

By comparison, in the world of constant-depth circuits we do have strong lower-bounds. H{\aa}stad's celebrated Switching Lemma gives a sharp lower bound of $2^{\Omega(n^{1/(d - 1)})}$ for depth-$d$ circuits with gates of unbounded fan-in that compute the Parity function. This remains the best known lower bound against $\AC_0$ for any explicit function. Following Razborov and Rudich, it can be shown that under a suitable cryptographic assumption, natural proofs cannot prove a significantly better lower bound, such as of $2^{n^{c/d}}$ for some large constant $c$.

In this paper we revisit the natural-proofs barrier from an \emph{unconditional} perspective. We focus on $\AC_0$-natural proofs, namely natural proofs whose associated distinguisher is itself computable by an $\AC_0$ circuit. Razborov and Rudich already observed that lower bounds based on the Switching Lemma are $\AC_0$-natural. We begin by showing that this is also true for most known lower-bound techniques against constant-depth circuits.

We then establish an \textit{unconditional} barrier for such proofs. By localizing the Trevisan--Xue pseudorandom generator, we construct a pseudorandom function computable by depth-$d$ circuits of size $2^{n^{O(1/d)}}$, that unconditionally fools $\AC_0$ distinguishers of arbitrary constant depth. It follows, for example, that no $\AC_0$-natural proof can prove a lower bound
greater than $2^{n^{7/(d-5)}}$ against depth-$d$ circuits.

The result is not tight: the ideal barrier would match the $2^{\Theta(n^{1/(d-1)})}$ Switching-Lemma frontier, instead of having $7/(d-5)$ in the power of $n$. Nevertheless, it gives an unconditional natural-proofs barrier in the same qualitative regime. The proof has a striking self-referential aspect: the
proof of security of the Trevisan--Xue generator crucially relies on the Switching Lemma, and so the Switching Lemma has been used to show that $\AC_0$-natural proofs, such as the Switching Lemma itself, cannot prove $\AC_0$ lower bounds significantly better than the state of the art.

\end{abstract}

\newpage
{%
\tableofcontents
}

\newpage

\section{Introduction}

\paragraph{The frontier for constant-depth circuit lower-bounds.}

A central problem in circuit complexity is to prove lower bounds for explicit functions against constant-depth circuits. By constant-depth circuits we mean ``$\AC_0$'' circuits, i.e., circuits with unbounded-fan-in AND and OR gates, together with NOT gates. For depth-$d$ circuits, Håstad's Switching Lemma \cite{hastadswitching} gives lower bounds of size $2^{\Omega(n^{1/(d-1)})}$ for the Parity function, and this bound is tight for Parity. The challenge is to prove a stronger lower bound for some other explicit function.

Even for depth $3$, improving the known frontier is a central open problem \cite{hastad1993top,GurumukhaniPPST24}. Because of a depth reduction result by Valiant \cite{Valiant77}, a sufficiently strong depth-$3$ lower bound, namely a lower bound of size $2^{\omega(n/\log\log n)}$, would imply a super-linear lower bound for log-depth circuits. Thus progress on depth-$3$ lower bounds would have consequences beyond constant-depth complexity itself.

\paragraph{Natural proofs.}
We study this frontier through the lens of the Razborov--Rudich natural-proofs paradigm \cite{RAZBOROV199724}. Informally, a lower-bound proof, showing that a function $F$ is ``hard'', is called a \textit{natural proof}, if it identifies a property $\Phi(\cdot)$ of truth tables satisfying the following: 
\begin{enumerate}
    \item \emph{Usefulness}: $\Phi(F) = 1$ but $\Phi(f) = 0$ for all ``easy'' functions $f$.
    \item \emph{Largeness}: $\Phi(f) = 1$ for a noticeable fraction of all functions.
    \item \emph{Constructiveness}: $\Phi(f)$ can be efficiently computed when given the truth table of $f$ as input. 
\end{enumerate}

Hence, a useful, large, constructive property $\Phi$ is an efficient algorithm that distinguishes a random function from every ``easy'' function. Razborov and Rudich have shown that all circuit lower-bound proofs known at the time are natural proofs, and hence all such lower-bound proofs yield efficient distinguishers. The key observation of Razborov and Rudich was then the following: under standard cryptographic assumptions, we can construct \emph{pseudorandom functions} (PRFs), namely, distributions over easy functions whose truth tables \textit{cannot} be distinguished from random truth-tables by any efficient distinguisher. It follows that, under standard cryptographic assumptions, natural proofs cannot prove a lower bound against, say, polynomial-size circuits.

Notice here that this conclusion\,---\,better lower-bounds cannot be proven by natural proofs\,---\,only follows conditionally: it follows under a certain, unproven cryptographic assumption.

For instance, under the widely believed assumption that sufficiently secure PRFs are computable by polynomial-size log-depth ($\NC_1$) circuits, the usual natural-proofs barrier already rules out natural proofs of stronger $\AC_0$ lower bounds. This is because of a standard simulation of $\NC_1$ by depth-$d$ circuits, which implies that $\NC_1$-PRFs of depth $c \log n$ can also be computed by depth-$d$ circuits of size $2^{n^{c/(d-1)}}$. Consequently, \uline{conditional} on this cryptographic assumption, there is a constant $c$ such that no natural proof can prove depth-$d$ lower bounds above $2^{n^{c/(d-1)}}$.\footnote{Constructions of candidate PRFs in depth $(1+\epsilon)\log n$ can be found in the literature~(see for example \cite{FanL022}), but these are PRFs with a different notion of security. They rely on the distinguisher only being able to make polynomially many queries to the truth table, and our distinguishers are allowed to read the whole truth table.}  However, this does not explain the exact frontier suggested by the Switching Lemma.

\paragraph{$\AC_0$-natural proofs and an unconditional natural-proofs barrier.}

Razborov and Rudich have also observed that lower bounds against $\AC_0$ circuits have an additional feature: the associated property $\Phi$ is not merely efficiently decidable, it is actually decidable by a small $\AC_0$ circuit. We may call such proofs \emph{$\AC_0$-natural}. One of our contributions, in \Cref{sec:all-natural}, is surveying the more recent lower-bounds and showing that this remains essentially true today: most of the known lower-bound proofs against $\AC_0$ circuits are in fact $\AC_0$-natural. In particular, \emph{lower-bound proofs based on the Switching Lemma are $\AC_0$-natural}.

Thus, in the setting of $\AC_0$ lower-bounds, the distinguishers arising from $\AC_0$-natural proofs are themselves $\AC_0$ circuits. And since unconditional pseudo-random generators fooling $\AC_0$ are known (see \cite{HatamiH24} for a survey of these results), one can speculate whether there exists an unconditional version of the natural-proofs barrier.

Concretely, to rule out $\AC_0$-natural proofs beyond the known depth-$d$ frontier, one would have to construct pseudo-random functions satisfying two requirements: they should be computable by depth-$d$ circuits of size near $2^{n^{{1}/{(d-1)}}}$, and they should fool $\AC_0$ distinguishers (of size $\poly(2^n) = 2^{O(n)}$, or ideally even $2^{n^{O(1)}}$) that inspect the full truth table.\footnote{One should contrast this kind of complexity-theoretic PRF with a cryptographic PRF. In \cite{degwekar2016fine}, the authors construct a PRF computable in $\AC_0$ which fools $\AC_0$ distinguishers that are given a polynomially-long sequence of randomly-sampled input-output pairs $(x, f(x))$. Such distinguishers are exponentially weaker than the distinguishers we aim to fool: in our case the $\AC_0$ distinguisher has access to \textit{the entire truth-table}, not just a small number of samples.} This leads to the following question.

\begin{question}
\label{qu:prf}
For every $d \ge 3$ does there exist a distribution on functions each computable by depth-$d$ circuits of size $2^{O(n^{{1}/{(d-1)}})}$ that fools $\poly(2^n)$-size $\AC_0$ circuits?
\end{question}

Note that an affirmative answer implies an \uline{unconditional} natural-proofs barrier in the following sense. Take any property $\Phi$ of $n$-variate Boolean functions that rejects all functions computable by depth-$d$ circuits of size $2^{O(n^{1/(d - 1)})}$ while accepting a random function with high probability. Then $\Phi$ itself cannot be computed by an $\AC_0$ circuit of size $\poly(2^n)$. This would unconditionally explain why techniques such as the Switching Lemma cannot be used to yield stronger lower bounds than the current state of the art.

It should be noted here that Razborov and Rudich ~\cite{RAZBOROV199724} have shown an unconditional natural-proofs barrier, specifically showing that no polynomial size $\AC_0$-natural proof can prove a superpolynomial lower bound against polynomial size $\AC_0[2]$ circuits. This is based on the Nisan--Wigderson generator~\cite{NisanW1994} that is implementable as a PRF using constant depth polynomial size $\AC_0[2]$ circuits and that fools $\AC_0$ circuits. More specifically there is a \emph{fixed} depth $d$ such that for any constant depth $d'$ and polynomial (in $2^n$) size $S'$, there is a $\AC_0[2]$ PRF of depth $d$ and size polynomial in $n$ that fools all depth-$d'$ size-$S'$ $\AC_0$ circuits. It is natural to ask whether this PRF is also implementable with efficient subexponential size $\AC_0$ circuits so that it may have some bearing on our question. The answer is an unsatisfactory yes. The Nisan--Wigderson generator is implementable via depth-$d$ size-$2^{n^{O(1/d)}}$ circuits. However, the constant in the double exponent now depends on the depth $d'$ that we are trying to fool, and this gives us a ``weak'' $\AC_0$-natural-proofs barrier.

We go over a similar case in more detail: In \Cref{sec:k-wise-independence}, we will show a very similar weak natural-proofs barrier arising from a simple PRF, whose security follows from Braverman's theorem. This PRF has better parameters than the $\AC_0$ PRF we just sketched (which is based on $\AC_0[2]$ PRFs). It will be used in \Cref{sec:unconditional} as an important step in our main result, which is a stronger $\AC_0$-natural-proofs barrier.

\paragraph{Main result.}

Our first observation is that Braverman's theorem \cite{Braverman}, which says that polylogarithmic-wise independence fools $\AC_0$, already yields a simple pseudo-random function distribution supported on functions computable by circuits of depth $d$ and size $2^{n^{O(1/d)}}$. As alluded to above, the constant in the double exponent depends on the depth $d'$ of the $\AC_0$ circuits that the PRF fools. This gives us the following ``weak'' unconditional natural-proofs barrier.

\begin{theorem}[``Weak'' Natural-Proofs Barrier (informal, see \Cref{thm:weak-unconditional-barrier})]
\label{thm:intro-npb-braverman}
Let $d \geq 5$, and $\Phi : \Q^{2^n} \mapsto \Q$ be a natural property constructed in polynomial size and constant depth $d_\Phi$. Then $\Phi$ cannot show a lower bound of $2^{n^\eps}$ for depth $d$ circuits unless $\eps < (3d_\Phi+7)/(d-3)$, i.e. unless $d_\Phi$ is sufficiently large as a function of $d$, or $\eps$ is sufficiently small.
\end{theorem}

Our formal version of the above theorem also generalizes it to properties constructed in quasipolynomial size and constant depth. Note that the above rules out a depth-2 poly-size natural property from being able to prove that a depth-$20$ circuit requires size at least $2^{n^{14/17}}$ but it does not rule out a depth-$3$ poly-size natural property from doing the same.\footnote{What it does rule out is a property like the one from the Switching Lemma, which was always a poly-size property of a \emph{fixed} depth, from simultaneously proving a $2^{n^{\omega(1/d)}}$ lower bound on depth $d$ circuits for all large enough $d$.}

This construction is nevertheless useful as a warm-up (it appears in \Cref{sec:k-wise-independence}) and the PRF it constructs is useful for our next result.
Our next result removes the dependence between $d_\Phi$ and $d$, by analyzing the PRG construction by Trevisan and Xue \cite{trevisan2013derandomized}, which is based on an earlier construction by Ajtai and Wigderson \cite{ajtai1985deterministic}.

Our main technical contribution is to ``localize'' the Trevisan--Xue PRG in constant depth.\footnote{See the discussion in the ``Connections with MCSP'' section of the Introduction for a comparison with another work that also uses a localization (but not in constant depth) of the Trevisan--Xue switching lemma.} In their PRG, the pseudorandom output string is, say, an $N$-bit string which is output by a circuit of size comparable to $N$, when given a seed of length $\approx \mathsf{polylog}(N)$. Since we need a PRF, we need the $N$-bit string which is output to be the truth-table of an $\AC_0$ circuit $C$ of depth $d$ and size $2^{n^{O(1/d)}}$. In other words, we must generate a pseudorandom circuit $C$ and then the $i$-th bit of the pseudorandom output string is $C(i)$. Indeed this can be done, and it yields a pseudo-random function distribution fooling distinguishers of arbitrary constant depth. This gives us our stronger natural-proofs barrier.

\begin{theorem}[Natural-Proofs Barrier (informal, see \Cref{thm:unconditional-natural-proofs})]
\label{thm:intro-npb}
Let $d \geq 7$, and $\Phi : \Q^{2^n} \mapsto \Q$ be a natural property constructed in polynomial size and constant depth $d_\Phi$. Then $\Phi$ cannot show a lower bound of $2^{n^\eps}$ for depth $d$ circuits, unless $\eps < 7/(d-5)$.
\end{theorem}

Again, our formal version of the above theorem generalizes it to properties constructed in quasipolynomial size as well. Note that we can now rule out \emph{any} constant depth poly-size natural property from being able to prove that a depth-$20$ circuit require $2^{n^{7/15}}$ size. In fact, even $\log \log$-depth poly-size natural properties cannot prove it. The lower bound of $2^{n^{8/(d-1)}}$ is \uline{unconditionally} out of reach for $\AC_0$-natural properties for all $d>33$.

While we consider this to be a very interesting barrier, the theorem falls short of the ideal parameters in \Cref{qu:prf} by a constant factor in the exponent. But notice that this constant factor in the exponent is exactly the same regime that follows from the assumption that there exist sufficiently secure $\NC_1$ PRFs, and here we have proven it \uline{unconditionally}, even though we only fool $\AC_0$ distinguishers.

\paragraph{Limits of the Switching Lemma via the Switching Lemma.} 

The proof has a notable self-referential feature.  The centerpiece of the PRG of Ajtai--Wigderson is a repeated application of the Switching Lemma. The Trevisan--Xue improvement, which we do need, uses a derandomization of the Switching Lemma. And so we have used the Switching Lemma to show, for example, that the Switching Lemma cannot prove Valiant's desired lower-bound. We emphasize: the Switching Lemma, which is used to prove lower-bounds against constant-depth-circuits, has here been used to formally show that the Switching Lemma cannot prove better lower-bounds.

\paragraph{The natural-proofs barrier as a self-referential obstacle.}

This result gives a concrete instance of a broader intuition surrounding the natural-proofs barrier, that it reveals a self-referential obstacle against proving lower-bounds (see \cite{aaronson2003p,mulmuley2010explicit}). If one informally thinks of natural proofs as a proof system, the natural-proofs barrier could be written as:
\begin{align*}
\mathsf{CryptoIsPossible} \implies \mathsf{NaturalProofs} \nvdash \text{``}\PTIME \neq \NP\text{''}.
\end{align*}

Although $\mathsf{NaturalProofs}$ is not a formal proof system, this slogan captures the sense in which the barrier resembles a self-obstruction phenomenon: Under cryptographic assumptions, natural proofs cannot show that $\PTIME \neq \NP$, and so in particular they cannot prove those same cryptographic assumptions. Informally, natural proofs fail because the same computational hardness that one hopes to prove can give rise to pseudorandom objects that evade the proof technique.

In the Razborov--Rudich setting this obstruction is conditional on cryptographic assumptions. In our setting, for $\AC_0$-natural proofs against constant-depth circuits, the obstruction is unconditional, and the self-referential aspect is especially visible. The Switching Lemma is the main tool behind the known lower bounds against constant-depth circuits. Yet, through the Ajtai--Wigderson and Trevisan--Xue pseudorandom generators, the Switching Lemma also helps construct the pseudorandom functions that rule out stronger lower bounds by $\AC_0$-natural proofs. Thus the same method that proves the known lower bounds also explains why that family of methods cannot prove stronger lower bounds.

\paragraph{Related work.}
The explanatory power of the natural-proofs barrier applies to different computational models where lower-bound provers are stuck, not just Boolean circuits and formulas. For example, a natural-proofs barrier can be established also in the context of algebraic complexity \cite{grochow2017towards} (although some subtlety is needed in this context and the barrier is less certain \cite{chatterjee2020existence}) and data-structure lower-bounds in the cell-probe model \cite{korten_pitasi_impagliazzo_2025, KLMS26}. We seem to be witnessing a repeated, reliable appearance of the natural-proofs barrier along the entire lower-bounds frontier.

The primary implication of the natural-proofs barrier is that it rules out natural proofs of lower bounds against strong complexity classes that contain sufficiently secure PRFs. Thus, for example, since it is believed that $\TC_0$ contains such PRFs, we do not expect that we can prove lower bounds against this class by natural proofs. However, the barrier does not seem to address why even for weaker classes we are stuck. For example, we have failed to prove super-quadratic lower-bounds against branching programs since Ne\v ciporuk in 1966 \cite{nechiporuk1966,razborov1991lower}, or super-cubic lower-bounds against Boolean formulas \cite{S61,K72,A87,IN93,PZ93,H98,T14}, or super-linear lower-bounds against Boolean circuits (we are stuck at $5n$ \cite{lachish2001explicit,iwama2002explicit}). It is conceivable that there are PRFs precisely at these thresholds. Along these lines, Raz \cite{raz2026note} recently showed that $n\cdot (\log n)^{\omega(1)}$ (``strongly super-linear'') lower-bounds against (Boolean or algebraic) circuits computing linear functions cannot be proven with natural proofs under strong but reasonable cryptographic assumptions. Fan, Li, and Yang \cite{FanL022} considered circuit constructions of PRFs which have oracle access to the function. They used this to show that under suitable assumptions, certain black-box techniques cannot yield strong lower bounds.

\paragraph{Connections to \texorpdfstring{$\mcsp$}{MCSP}.}

For a circuit class $\calC$ the Minimum Circuits Size Problem with respect to $\calC$ ($\mcsp_{\calC}$) asks, given the truth table of a function $f$, whether $f$ can be computed by a circuit in $\calC$. When $\calC$ is the class of general circuits of size $o(2^n / n)$ we drop it from subscript. $\mcsp$ is easily shown to be NP and it is commonly conjectured to be NP-hard. Its NP-hardness is only known in restricted cases \cite{HiraharaOS18,Ilango24} or under complexity assumptions \cite{HiraharaI25}. There are also unconditional lower bounds for MCSP against various restricted circuit classes (see, e.g.,  \cite{HiraharaS17,CheraghchiKLM20}).

Observe that $\neg \mcsp_{\calC}$ is a \emph{large} and \emph{useful} property against $\calC$. Furthermore, any lower bound for $\mcsp_{\calC}$ shows the extent of non-constructiveness of this property. Natural-proofs barriers say that \emph{every} large and useful property must be non-constructive. So the direction that we propose and prove results in also proves lower bounds for $\neg \mcsp_{\calC}$.

Specifically for a size parameter $s = 2^{n^{O(1/d)}}$, consider the class $\calC$ of depth-$d$ size-$s$ $\AC_0$ circuits. An implication of our results is that $\mcsp_{\calC}$ is not in poly-size $\AC_0$. We note the two most closely related results we found in the literature. Cheraghchi et al.~\cite{CheraghchiKLM20} study the $\AC_0$ constructiveness of $\mcsp$ and prove a much stronger lower bound, but this has no direct implication on the $\AC_0$ constructiveness of $\mcsp_{\calC}$. Ilango~\cite{Ilango24} studies $\mcsp_{\calC}$ where $\calC$ is $\AC_0$ formulae, showing that it is NP-hard under quasipolynomial-time randomized reductions, but says nothing of the circuit variant that our paper focuses on.

The paper by Cheraghchi et al.~\cite{CheraghchiKLM20} is also relevant because their result about the $\AC_0$ constructiveness of $\mcsp$ lower bounds uses localizations of PRGs, but in an entirely different manner. Their idea is to restrict the input truth table of MCSP. With the restriction being easy to compute \emph{locally by a small circuit (with no depth restrictions)} it follows that there is a small circuit whose truth table agrees with said restriction. However the number of truth tables that agree with the restriction is too large for them all to be easy, so the restriction cannot make MCSP constant. Since such restrictions exist that make any $\AC_0$ circuit constant, their lower bound on MCSP follows. Interestingly they use Trevisan--Xue's derandomization of the switching lemma to generate such restrictions locally.

\section{Preliminaries}

\subsection{Notation}

We start with some basic notation. We denote the set of all Boolean functions in $n$ variables by $\cF(n)$, and the set of all Boolean functions in $n$ variables computable by depth-$d$ circuits of size $S$ by $\cF({n, d, S})$.

\begin{definition}[Partial Assignment]
    A partial assignment on $N$ bits is represented as $X \in \{0,1,\ast\}^N$. Indices $\{u : X_u = \ast\}$ are the free indices of $X$ and the rest are the fixed indices.
\end{definition}

Partial assignments are closely related to restrictions.

\begin{definition}[Restriction]
    A restriction is represented as $\rho : [N] \to \{0,1,\ast\}$. The free indices of $\rho$ is the set $\rho^{-1}(\ast)$ and the rest are the fixed indices. We also represent a restriction as $I \mapsto \alpha$ where $I$ is the set of fixed indices of the restriction and $\alpha \in \ZO^{|I|}$ are the values set in the fixed indices. A restriction can act on various objects:
    \begin{itemize}
        \item Composition of two restrictions on $N$ bits: Applying restrictions $\rho_0$ and then $\rho_1$ is equivalent to applying their composed restriction $\rho_0\rho_1$ where $$\rho_0\rho_1(u) = \begin{cases}
            \rho_0(u) & \text{if }\rho_0(u) \in \ZO \\
            \rho_1(u) & \text{otherwise.}
        \end{cases}$$

        \item Given a partial assignment $X \in \{0,1,\ast\}^N$, applying $\rho$ to $X$ would produce the partial assignment $X'$ where
        $$X'_u = \begin{cases}
            X_u & \text{if }X_u \in \ZO \\
            \rho(u) & \text{otherwise.}
        \end{cases}$$
        Note that $\rho_0\rho_1$ applied to $X$ is the same as applying $\rho_0$ to $X$ and then applying $\rho_1$ to the result.
        
        \item For a function $\Phi : \ZO^N \to \ZO$, restricting with $\rho := I \mapsto \alpha$ yields the restricted function $\Phi|_{\rho} : \ZO^{N - |I|} \to \ZO$ defined as $\Phi|_{\rho}(X) = \Phi(X')$ where $X'$ is equal to $\alpha$ in indices $I$ and equal to $X$ on the other indices.
    \end{itemize}
\end{definition}

\subsection{What is a natural proof?}

The notion of a \textit{natural proof} is not a formal notion: there is no rigorous definition, and this is for good reason (as we shall see). Instead, the formal notion is that of a natural property.

\begin{definition}[$\AC_0$-natural properties] Let $n,d,S,d_\Phi,S_\Phi \in \N$ and $p \in (0,1)$.
An $(n,d,S,d_\Phi,S_\Phi,p)$-natural property is a predicate $\Phi:\ZO^{2^n}\to\ZO$,\footnote{It should be understood that $d_\Phi$ and $S_\Phi$ are not somehow indexed by $\Phi$, it is just convenient to use the subscript $\Phi$ to reinforce that $d_\Phi$ and $S_\Phi$ are the depth and sizes of the distinguisher, as opposed to the depth $d$ and size $S$ of the function whose complexity is being lower-bounded.} such that:
\begin{itemize}
    \item (Usefulness) If $f \in \cF({n,d,S})$, then $\Phi(f) = 0$. I.e. the property rejects all truth tables of easy functions.
     \item (Largeness) If $f \in \cF(n)$ is chosen uniformly at random, then $\Pr_f[\Phi(f) = 1] \geq p$. I.e., the property is true of a random function with ``high'' probability $\ge p$.
    \item (Constructivity) $\Phi \in \cF({2^n,d_\Phi,S_\Phi})$. I.e., the property is itself easy.
\end{itemize}
If we omit $p$, then we are considering that $p = \Omega(1)$. In fact we will have $p$ close to $1$ for all natural proofs we will consider.
\end{definition}

Let us now argue why the definition of a natural proof should actually be an informal definition. Note that if we have an $(n,d,S,d_\Phi,S_\Phi)$-natural property $\Phi$ such that $\Phi(f) = 1$, this necessarily implies that the function $f$ does not have circuits of depth $d$ and size $S$. I.e., $\Phi(f) = 1$ gives us a lower-bound on the complexity of $f$, which follows from \textit{usefulness}. One might then be tempted to define a natural proof as a proof that explicitly constructs a natural property $\Phi$ and shows that $\Phi(f) = 1$. However, the reality is that most lower-bounds in the literature explicitly use specific properties of the function $f$, which are not true of a random function, or properties which appear to be hard to compute efficiently when given the truth-table of $f$. It would then seem that such lower-bounds are not \textit{natural proofs}, because the lower-bound-implying property does not obey largeness, or constructivity, hence it is not natural. 

The very surprising fact is that, whenever one finds such a seemingly-unnatural proof in the literature, then systematically, time and again, after a bit of thinking, one concludes that the proof can be changed only slightly, to extract from it a similar argument which does not actually require any of the specific properties of $f$ that were used in the original argument, and which indeed also applies to a random function. And we will find that this extracted property is very simple to compute when given $f$ as input. One has then obtained a natural property from the original lower-bound proof.

This process of extraction of a natural property is called \textit{naturalizing} a lower-bound proof. We may then think of it as follows: the proof did not \textit{seem} natural, but deep down it actually was a natural proof. So in practice we call any \textit{naturalizable} lower-bound proof a \textit{natural proof}. Of course, such a concept of natural proof is necessarily imprecise: there is no way to formally specify or even predict what will be necessary to do in order to naturalize a given lower-bound proof. Nonetheless, if a natural-proofs barrier exists, no lower-bound proof can exist which can be naturalized, no matter how one might try to do it. So this imprecise concept is somehow the right concept. We must therefore settle for an informal definition:

\begin{definition}[informal]\label{def:natural-proof} If we have a proof that a function $f:\ZO^n\to\ZO$ cannot be computed by depth-$d$, size-$S$ circuits, and from this proof we can somehow \emph{extract} an $(n,d,S,d_\Phi,S_\Phi,p)$-natural property $\Phi$, where $p$ is reasonably large (e.g., $p = \Omega(1)$), $d_\Phi$ is a constant and $S_\Phi$ is not too large (e.g. $S_\Phi \le |f|^{O(1)} = 2^{O(n)}$, or maybe $S \le 2^{(\log |f|)^{O(1)}} = 2^{n^{O(1)}}$), and the proof can be adapted to show that $\Phi(f) = 1$, then we call this proof an \textit{$\AC_0$-natural proof}.
\end{definition}

\subsection{Pseudo-random functions (PRFs)}

\begin{definition}
    An \emph{$(n,d,S)$-function generator} is a distribution $\mu$ over  $\cF(n,d,S)$. In this paper our function generators will additionally be ``uniform''. That is, there is a single depth-$d$ size-$S$ circuit which has $n$ input wires and an additional number of seed wires with the distribution $\mu$ being generated by sampling the values of the seed wires uniformly at random and outputting the function (from $\{0,1\}^n$ to $\{0,1\}$) computed by the remaining circuit.
\end{definition}

\begin{definition}[$\epsilon$-fooling]
    Let us consider a function generator $\mu$ and let $\mathcal{A}$ be a class of algorithms, which we call \textit{distinguishers}. We say that $\mu$ \emph{$\epsilon$-fools} $\mathcal{A}$ iff:
    $$\forall A \in \mathcal{A}: \left\lvert\Pr_{f\sim\mu  }[A(f)] - \Pr_{f\sim \Unif }[A(f)]\right\rvert < \epsilon,$$
    where $\Unif$ denotes the uniform distribution.
\end{definition}

\begin{definition}
    An \emph{$(n,d,S,d_\Phi,S_\Phi,\epsilon)$-pseudorandom function generator} is an $(n,d,S)$-function generator that $\epsilon$-fools all distinguishers in $\cF(2^n,d_\Phi,S_\Phi)$.
\end{definition}

As noted in the introduction, the existence of PRFs in a circuit class precludes the existence of natural proofs against said circuit class. Here we state the instantiation of this logic relevant to our paper.

\begin{lemma}\label{lem:prfcorollary}
    Let $n,d,S,d_{\Phi},S_{\Phi}\in \N$ and $\epsilon \in (0,1)$. If there is an $(n,d,S,d_\Phi,S_\Phi,\epsilon)$-pseudorandom function generator, then there is no $(n,d,S,d_\Phi,S_\Phi,\epsilon)$-natural property.
\end{lemma}

\begin{proof}
    We will prove this by contradiction. Let $\Phi$ be an $(n,d,S,d_\Phi,S_\Phi,p)$-natural property and let $f_1,f_2 : \ZO^n \to \ZO$ be functions sampled from the PRF and from the uniform distribution respectively. Since $f_1 \in \cF(n,d,S)$, from the usefulness of $\Phi$ we know $\Phi(f_1)=0$ with probability $1$. On the other hand from the largeness of $\Phi$ we know that with probability at least $\epsilon$, $\Phi(f_2)=1$. Hence $\Phi$ is a distinguisher, and by constructivity it is in $\cF(2^n,d_{\Phi},S_{\Phi})$. This contradicts the definition of the PRF.
\end{proof}

We will use the following folklore proposition about PRFs, whose proof we include in Appendix \ref{sec:deferredproofs} for completeness.

\begin{restatable}{proposition}{fooling}
\label{prop:foolingmanipulations}
    Let $\mu_1$ be a distribution on $N_1$ bit strings that $\epsilon_1$-fools a class of algorithms $\mathcal{A}$ that is closed under restrictions and padding. Similarly define $\mu_2$ with parameters $N_2$ and $\epsilon_2$. Then we have the following observations.
    \begin{enumerate}
        \item For any $I \subseteq [N_1]$, the marginal of $\mu_1$ to the indices $I$ will also $\epsilon$-fool $\mathcal{A}$.
        \item The distribution $\mu_1 \times \mu_2$ will also $(\epsilon_1+\epsilon_2)$-fool $\mathcal{A}$. (In particular $\mu_1 \times \Unif_{N_2}$ continues to $\epsilon_1$-fool $\mathcal{A}$.)
    \end{enumerate}
\end{restatable}

\subsection{Multiplicative Approximate Counting}

We will need the following result, which is a consequence of a result of Ajtai and Ben-Or \cite{ajtai1984theorem},
to show that many known lower bound proofs in fact give $\AC_0$-natural properties. We provide a proof of it in Appendix \ref{sec:deferredproofs}.

\begin{restatable}{theorem}{approxcount}
\label{thm:approximate-counting}
    There exists a function $\mathsf{GapCount}^{\times}_{n,t,\eps} : \ZO^n \to \ZO$ satisfying
    \[
    \mathsf{GapCount}^{\times}_{n,t,\eps}(x) = \begin{cases}
        0 & \text{if } |x| \le t,\\
        1 & \text{if } |x| \ge \left(1 + \eps\right) \cdot t
    \end{cases}
    \] that is computable by a circuit of depth $O(1+\log (1/\eps)/\log \log n)$ and size polynomial in $n$.
\end{restatable}

\section{Most lower-bounds against \texorpdfstring{$\AC_0$}{AC0} are \texorpdfstring{$\AC_0$}{AC0}-natural}\label{sec:all-natural}

In this section we show that almost all known lower-bound techniques against constant-depth $\AC$-circuits are $\AC_0$-natural. We will cover lower-bounds based on the Switching Lemma (\Cref{sec:switching-lemma-is-natural}), lower-bounds bounds based on sensitivity and concentration of the Fourier spectrum (\Cref{sec:sensitivity-is-natural}), and top-down lower-bounds (\Cref{sec:HJP,sec:PPZ,sec:GRSS}). Among these some of the top-down lower-bounds are not $\AC_0$-natural, but for all except a couple of them we were able to find thematically similar $\AC_0$-natural proofs proving nearly the same lower bounds. It could be that they do, in fact, naturalize, and we simply did not figure out how to do it: we leave naturalizing these proofs as an \uline{open problem}.

\subsection{Bottom-up techniques}

These techniques generally analyze circuits in a bottom-up fashion, that is by starting from the input level and working the way up to the output gate.

\subsubsection{Proofs that use the Switching Lemma}\label{sec:switching-lemma-is-natural}

The first lower-bounds against $\AC_0$ circuits were proven by Ajtai \cite{Ajtai83} and Furst-Saxe-Sipser \cite{FurstSS84}, who showed that a function $f$ which is computed by a small CNF, becomes expressible by a small DNF with high probability over a random subcube $\rho$ of reasonably large dimension, so the order of the AND/OR can be switched. It follows from this Switching Lemma that any function $f$ computable by a small depth $d$ circuit can also be computed by a depth $d-1$ circuit inside the subcube. The parameters for this lemma were later improved by Yao \cite{yao1985separating}, and ultimately H{\aa}stad \cite{Hastad86}. From H{\aa}stad's Switching Lemma we can derive the following corollary:

\begin{corollary}
    If $f \in \cF({n,d,S})$, and $\rho$ is a $p$-random restriction with $p = \frac{1}{O(\log S)^{d-2}}$, then
    \[
    \Pr[f|_\rho\text{ can be computed by a $k$-CNF}] \ge 1 - o(1),
    \]
    where $k= 2 \log S$. In particular, $f$ has a monochromatic subcube of dimension $\Omega(p \cdot n) - k$. If $S< 2^{\Omega\left(n^{\frac{1}{d-1}}\right)}$, this dimension is at least $2 \log n$. 
\end{corollary}

From this Corollary it follows that any function $f$, such as parity, which is non-constant on every Boolean subcube of dimension $2\log n$ is not in $\cF({n,d,S})$. Given the truth table of $f$, i.e., given $f \in \ZO^N$, where $N=2^n$, we can test whether $f$ is non-constant in every subcube of dimension $s$ by an $\AC_0$ circuit of depth $d_\Phi = 2$ and size $S_\Phi =O(N^2)$. A random function will have this property with probability $1 - o(1)$. And so we have a $(d,S,d_\Phi,S_\Phi,o(1))$-natural proof. 

\medskip
That this proof is $\AC_0$-natural had already been observed by Razborov and Rudich. In the time since, many variants of the Switching Lemma have been proven \cite{beame1990lower,beame1994switching,Rossman08,impagliazzo2012satisfiability,haastad2014correlation,rossman2018average} For example, in \cite{Rossman08}, Rossman has shown that the function $\mathsf{CLIQUE}_k:\ZO^{\binom{n}{2}} \to \ZO$, which indicates whether a given undirected graph has a clique of (constant) size $k$, cannot be computed by constant-depth circuits of size $O(n^{\frac{k}{4}})$. The core technical lemma is the following:

\begin{theorem}
     Suppose $f_n \in \cF({\binom{n}{2},O(1), O(n^t)})$, i.e. $f_n$ is a function on $\binom{n}{2}$ bits, computed by depth $d = O(1)$ circuits of size $O(n^t)$, and $t > 1/2$. For any integer $k \ge 3$ and $\alpha = \frac{1}{2t-1}$, let $G = \mathsf{ER}(n, n^{-\alpha})$ be an Erd\"os-R\'enyi random graph, and let $A$ be a uniform random set of $k$ vertices of $G$. Then, for all sufficiently large $n$, $f_n(G) = f_n(G \cup K_A)$ with high probability over the choice of $G$ and $A$, where $K_A$ is the clique on $A$.
\end{theorem}

If $t \le \frac{k}{4}$, for the given choice of $\alpha$, we have that $\Pr[\mathsf{CLIQUE}_k(G) = 1] = o(1)$, and of course $\mathsf{CLIQUE}_k(G \cup K_A) = 1$ always, so $f_n$ does not compute $\mathsf{CLIQUE}_k$. More generally, $f_n$ will fail to compute any function $f$ such that $\Pr_{G,A}[f(G) = f(G \cup K_A)]$ is not close to $1$. We can then take the following as our natural property. $\Phi(f) = 1$ if and only if, for every graph $G$ without a $k$-clique, $\Pr_A[f(G) \neq f(G \cup K_A)] \ge \frac{1}{4}$. If $f$ is a random function, we expect about half of the choices $A$ to give $f(G) = f(G \cup K_A)$, and since there are at least $\binom{n}{3} \gg \binom{n}{2}$ such choices, by concentration of measure it is $\ll 2^{-\binom{n}{2}}$-likely that some graph $G$ will have $f(G) \neq f(G \cup K_A)$ for fewer than a $\frac{1}{4}$ fraction of the choices of $A$. We take a union bound and conclude that this property is large. We can check whether this property holds using a constant depth circuit of size $2^{O(\binom{n}{2})}$, by taking an AND over all $G$, and counting the number of $A$ such that $f(G) \neq f(G \cup K_A)$ (this is a small number, we can count it precisely). Then this property, and hence Rossman's proof, are $\AC_0$-natural.

\subsubsection{Proofs that use sensitivity and estimates on Fourier mass}\label{sec:sensitivity-is-natural}

The \textit{average sensitivity} of a function $f:\ZO^n\to\ZO$ is the expected number of sensitive coordinates: $S(f) = \frac{1}{2^n} \sum_{x\in\ZO^n} S(f,x)$, where $S(f,x) = \lvert\{i \in [n] \mid f(x) \neq f(x \oplus e_i) \}\rvert \in \{0, \ldots, n\}$. In \cite{Boppana97}, Boppana has proven:

\begin{theorem}\label{thm:boppana}
    If $f \in \cF({n,d,S})$, then $S(f) = O(\log S)^{d-1}$.
\end{theorem}

Since the sensitivity of parity is $n$, we get the same lower-bound of $S \ge 2^{\Omega(n^{\frac{1}{d-1}})}$ against parity. Using approximate counting (\Cref{thm:approximate-counting}), a circuit of constant depth and size $2^{O(n)}$ can compute $\sum_{x \in \ZO^n} S(f,x)$ up to an additive error of $\frac{2^n}{n^2}$. This is enough to estimate $S(f)$ up to a $o(1)$ additive error, and so any lower-bound proven from \Cref{thm:boppana} is $\AC_0$-natural.

A number of results starting with Linial, Mansour, and Nisan \cite{linial1993constant} show that constant-depth circuits can be well-approximated (in $\ell_2$-norm) by low-degree polynomials \cite{haastad2001slight,haastad2014correlation,impagliazzo2012satisfiability}. The strongest such result is by Tal \cite{tal2017tight}:

\begin{theorem}[\cite{tal2017tight}]\label{thm:tal}
    If $f \in \cF({n,d,S})$, then $$W^{\ge k}(f) =\sum_{T: |T| \ge k} \hat f(T)^2 \le 2 \cdot 2^{-k/O(\log S)^{d-1}},$$ where $\hat f(T) = \frac{1}{2^n}\sum_{x \in \ZO^n} (-1)^{f(x)} \cdot (-1)^{|x \cap T|}$ is the Fourier coefficient of $f$.
\end{theorem}

By choosing $k = O(\log S)^{d-1}$ sufficiently large, we conclude that the entire Fourier tail has mass less than, say, $\frac{1}{10}$. But parity has all its Fourier mass at level $k = n$, and so we get the same lower-bound of $S \ge 2^{\Omega(n^{\frac{1}{d-1}})}$ against parity.

Again, while we cannot compute the Fourier tail $W^{\ge k}(f)$ exactly, we can use approximate counting to approximate the mass of the Fourier tail. Indeed,
\[
W^{\ge k}(f) = \frac{1}{2^{2n}}\sum_{T:|T| \ge k} \sum_{x,y\in\ZO^n} (-1)^{f(x)} \cdot (-1)^{f(y)} \cdot (-1)^{|x\Delta y|}.
\]
We can construct a constant-depth circuit of size $2^{O(n)}$ which computes every term in the sum, and using multiplicative approximate counting we can estimate the sum of the positive and negative terms up to a multiplicative error of $1 \pm \frac{1}{n^r}$, subtract the two estimates and divide by $2^{2n}$ (which is an easy division, even for $\AC_0$), to get an estimate of $W^{\ge k}(f)$ up to a multiplicative error of $O(1 \pm n^{-r})$. (Observe that here an additive approximation is not enough.) This is enough precision to distinguish the case when $f$ is computed by a simple constant-depth circuit, who will have Fourier tail $\le \frac{1}{10}$ versus a random function whose Fourier spectrum is roughly uniformly distributed, hence will have Fourier tail $\ge 1 - o(1)$.

\subsection{Top-down techniques}

Such techniques generally analyze circuits in a top-down fashion, that is by starting from the output gate, working the way down to input gates until reaching a contradiction.

\subsubsection{Proofs based on \texorpdfstring{$k$}{k}-limits}\label{sec:HJP}

Here we primarily focus on the the depth-$3$ lower bound technique of Håstad, Jukna and Pudl{\'a}k \cite{hastad1993top}. We will note that their property is not a natural property because it is small, and we will mention an $\AC_0$-natural property that is similar both in principle and in getting nearly the same lower bounds. We go into more detail about these properties in Appendix \ref{sec:topdownproperties}.

Håstad et al.'s lower bound technique works as follows. From a small $\Sigma_3$ circuit $C$ (i.e., depth 3 with a top $\mathsf{OR}$ gate) that computes $f$ one can easily find a CNF $D$ that rejects every input in $f^{-1}(0)$ and accepts every input in a large subset $T \subseteq f^{-1}(1)$. Now consider an input $y$ such that for every $Q \subseteq [n]$ of size at most $k$, there is an $x \in T$ such that $x_Q = y_Q$. Such a $y$ is called a \textbf{$k$-limit} of the set $T$. If $D$ has bottom fan-in $k$ then from the perspective of each clause in $D$, $y$ looks exactly like an element $x \in T$, so $y$ will also be accepted by all the clauses and hence by $D$. Hence such a $k$-limit of $T$ cannot exist in $f^{-1}(0)$.

The actual task is then to come up with what we will call ``limitful'' functions, i.e. functions $f$ such that in every partition of $f^{-1}(1)$ into at most $\ell$ parts, at least one part has a $k$-limit in $f^{-1}(0)$. Such functions would require $\Sigma^k_3$ circuits ($\Sigma_3$ with bottom fan-in $k$) of size at least $\ell$. The property of being limitful is far from constructive, although it is large and useful.

Håstad et al. use Ramsey theoretic arguments similar to those used in the sunflower lemma to show that if $f$ is a threshold function it is limitful. The property they use of $f$ is small (and hence not natural), and this seems to be needed for the Ramsey theoretic arguments to work. With this they show that the Threshold$_{n,n/k}$ function (that answers ``Does the $n$-bit input have at least $n/k$ ones?'') requires $e^{(n/k) \cdot (1-o(1))}$ size $\Sigma_3^k$ circuits. A lower bound of $e^{(n/2k) \cdot (1-o(1))}$ also follows for Majority since it embeds the above functions on $n/2$ bits.

(For the $k=3$ case specifically, a better lower bound of $\approx e^{1.35 n/2k}$ for Majority is known~\cite{GurumukhaniPPST24}. This property is even smaller and the analysis does not even involve $k$-limits. This is detailed later on in this section.)

For $\Sigma_3$ lower bounds they use an initial restriction to the input that makes the bottom fan-in bounded. The restriction turns a threshold function into a new threshold function, and so the $\Sigma^k_3$ lower bounds can be used. With this they get a lower bound of  $\approx 2^{0.849\sqrt{n}}$. The property that $f$ needs for this restriction to also be a part of the lower bound proof would be even smaller, although it would still be computable and useful.

\paragraph{A Natural Variant.} Meir and Wigderson~\cite{meir2019prediction} came up with a different way of finding $k$-limits. They showed that for any set $X$ of size at least $2^{n-o(n/k)}$, almost all its neighbours are $k$-limits. That is, if we choose $x \in X$ and $i \in [n]$ uniformly at random, then $x \oplus e_i$ is a $k$-limit of $X$ with probability at least $1-o(1)$. This provides a new way to generate limitful functions. Let $f$ be any function with a large number of $x \in f^{-1}(1)$ having high sensitivity. This yields an $\AC_0$-natural property that proves lower bounds against $\Sigma^k_3$ circuits. The large number of $k$-limits that it guarantees allows us to also use this property as is to prove $\Sigma_3$ lower bounds (see Appendix \ref{sec:topdownproperties} for details). This $\AC_0$-natural proof gives a $\Sigma^k_3$ lower bound of $2^{(1-o(1))n/2(k+1)}$, and a $\Sigma_3$ lower bound of $2^{(1-o(1))\sqrt{n}/2}$.

\paragraph{One lower bound that we could not yet naturalize.}\label{sssec:onelb}

Håstad et al. \cite{hastad1993top} also show a $\Sigma_3$ lower bound of $2^{\Omega(\sqrt{n})}$ on the function $\mathrm{AND}_{\sqrt{n}} \circ \mathrm{OR}_{\sqrt{n}} \circ \mathrm{AND}_2$. They prove this via the same technique covered above, showing that it is limitful. However in this setting the natural property used by Meir and Wigderson does not apply and hence this separation of $\Sigma_3$ vs $\Pi_3$ complexity does not have a natural proof yet. We leave it as an open problem to find an $\AC_0$-natural lower bound of $2^{\Omega(\sqrt{n})}$ for the above function.

\subsubsection{Proofs based on the Satisfiability Coding Lemma and the structure of \texorpdfstring{$k$}{k}-CNF solutions.}\label{sec:PPZ}

Paturi, Pudl{\'a}k and Zane  \cite{PaturiPZ99} proved a combinatorial lemma regarding $k$-CNFs which they used to devise a non-trivial algorithm for satisfiability of $k$-CNFs, and to show a tight $\Omega(n^\frac{1}{4} 2^{\sqrt n})$ lower-bound on the size of depth-3 circuits that compute Parity. The main component of their result is the following lemma. For a set $S \subseteq \Q^n$ we say that $x \in S$ is \emph{$j$-isolated} if there are at least $j$ coordinates such that flipping $x$ in any one of those coordinates results in an assignment outside $S$.

\begin{lemma}[Satisfiability Coding Lemma \cite{PaturiPZ99}]
Let $F$ be an $n$-variate $k$-CNF. For every $j \in [n]$, the number of $j$-isolated satisfying assignments of $F$ is at most $2^{n - j/k}$.
\end{lemma}

A  $\Sigma^k_3$ circuit is a  $\Sigma_3$ circuit such that every gate in the bottom layer has fan-in at most $k$.
Since a $\Sigma^k_3$ circuit, in particular, gives a covering of $j$-isolated accepting inputs of a function, the following general lower bound is an immediate corollary.

\begin{corollary}
Let $f$ be an $n$-variate function which accepts $S$ $j$-isolated assignments. Then any $\Sigma^k_3$ computing $f$ has size at least $S/2^{n - j/k}$. In particular,
\begin{enumerate}
    \item any $\Sigma^k_3$ circuit computing Parity has size at least $2^{\frac{n}{k} - 1}$ \cite{PaturiPZ99},
    \item a function $f$ sampled uniformly at random requires $\Sigma^k_3$ circuits of size $2^{\frac{n}{2k} - o(n)}$.
\end{enumerate}
\end{corollary}

The latter follows by observing that a random function has $2^{n - o(n)}$ $(\frac{n}{2} - o(n))$-isolated accepting inputs with high probability. Using the approximate counting circuit of \Cref{thm:approximate-counting}, we obtain an $\AC_0$-natural property useful against $\Sigma^k_3$ circuits of size $2^{n/2k - o(n)}$. Therefore, the weaker PPZ lower bound is an $\AC_0$-natural proof.

\medskip
Paturi, Saks, and Zane \cite{PaturiSZ00} showed a near-maximal lower bound for $\Sigma^2_3$ circuits computing \emph{affine dispersers}. An affine disperser for dimension $d$ is a function that is not constant under any affine space of dimension $d$. They showed that any $\Sigma^2_3$ circuit that computes an affine disperser for dimension $o(n)$ has size at least $2^{n - o(n)}$. No such lower bound is known for $\Sigma^k_3$ where $k \ge 3$. It is easy to see that a random function is affine disperser for dimension $O(\log n)$. Furthermore, the property of being an affine disperser for dimension $d$ can be trivially checked by a depth-3 circuit of quasipolynomial size for $d = O(\log n)$: exhaustively go through all affine spaces of dimension $d$ and check whether the function is constant in that affine space. It follows that the lower bound of \cite{PaturiSZ00} is a quasipolynomial-$\AC_0$-natural proof.

\medskip
Paturi, Pudl{\'a}k, Saks and Zane \cite{PaturiPSZ05} strengthened the Satisfiability Coding Lemma as follows to derive improved $k$-SAT algorithms as well as depth-3 circuit lower bounds.

\begin{lemma}
Let $F$ be an $n$-variate $k$-CNF. Then the number of satisfying assignments of $F$ which are of Hamming distance at least $\Delta$ from any other satisfying assignment is at most $2^{(1 - \frac{\pi^2}{6k} + O(\frac{1}{\Delta k}))n}$.
\end{lemma}

\begin{corollary}
    If $f$ is the characteristic function of an error correcting code of distance $\omega(1)$ accepting at least $S$ inputs, then it requires $\Sigma^k_3$  circuits of size at least $S/2^{(\pi^2/6k - o(1))n}$.
\end{corollary}

To see why this is a natural proof requires some subtlety. Clearly a random function will not be an error correcting code, not even of distance $2$, since the $2$-Hamming ball already has $n^2$ points, and so no accepting input is isolated. But if we choose a \textit{sparse} random function, namely, every output $F(x)$ is chosen independently at random with bias $\Pr[F(x)=1] = n^{-\omega(1)}$, then with high probability $F^{-1}(1)$ is an error-correcting code with distance $\omega(1)$, which accepts $2^{n-o(n)}$ inputs. Therefore the PPSZ proof is natural, in the sense that it produces a distinguisher which rejects all easy functions (with small $\Sigma^k_3$ complexity) from sparse random functions chosen with this bias. Moreover, as argued below, this notion of natural is just as relevant to the $\AC_0$-natural-proofs barrier.

We can easily convert an $\AC_0$ PRF for unbiased functions into a PRF $F:\ZO^n\to\ZO$ with any bias which is an inverse power of two. For example, if we want $\Pr[F(x)=1] = 2^{-q}$, we start by taking a PRF for a slightly larger bitlength $F':\ZO^{n + \lceil \log q  \rceil} \to \ZO$. We then use an AND to get $F$: $F(x) = \bigwedge_{x' \in [q]} F'(x x')$. It now follows that $F$ is sampled among random functions with bias $2^{-q}$, as promised. An $\AC_0$ distinguisher to break $F$ can be converted to an $\AC_0$ distinguisher to break $F'$ by composing it with the AND circuits that turn $F'$ into $F$ (i.e. convert the input function to a biased function as above and then use the distinguisher for $F$). This increases the depth of the distinguisher by at most $1$ and the size by at most $q$. Therefore, any (unbiased) $\AC_0$-natural-proofs barrier also gives us a barrier for biased $F$.

\paragraph{A potential non-natural lower-bound for Majority}\label{sec:maybe-not-natural}

The Majority function is a natural candidate for an explicit function that requires depth-3 circuits of size $2^{\omega(\sqrt{n})}$, i.e., beyond the current barrier \cite{hastad1993top}. Gurumukhani et al. \cite{GurumukhaniPPST24} recently proposed to investigate $\Sigma^k_3$ complexity of majority as an approach to determine its unrestricted depth-3 complexity. As the first non-trivial case they considered $k = 3$ and proved new lower bounds on the size of $\Sigma^3_3$ circuit computing the majority function. Their main result is the following combinatorial lemma.

\begin{lemma}
Let $F$ be an $n$-variate 3-CNF which rejects every assignment of Hamming less than $\lfloor n/2 \rfloor$. Then the number of satisfying assignments of $F$ of Hamming weight exactly $\lfloor n/2 \rfloor$ is at most $1.598^n$.
\end{lemma}

The following lower bound is immediate.

\begin{corollary}\label{cor:majority-lb}
Let $f$ be any $n$-variate function that rejects all inputs of Hamming weight less than $\lfloor n/2 \rfloor$ and accepts at least $S$ assignments of Hamming weight $\lfloor n/2 \rfloor$. Then any $\Sigma^3_3$ circuit computing $f$ has size at least $S / 1.598^n$.
\end{corollary}

There are two main issues that make it difficult to render the technique of \cite{GurumukhaniPPST24} as natural. Firstly, the property given in \cref{cor:majority-lb} is not large; in fact with high probability a random function does not satisfy this property. Secondly, we can indeed artificially condition on rejecting small Hamming weight inputs when sampling a random function. Under this condition we would only need to approximately count the number of accepting inputs in the middle slice which can be done in $\AC_0$. Therefore, in this artificial sense, the proof is natural. However, unlike in the case of sampling a sparse random function, as we did in the previous two sections, it is not clear how to construct PRFs supported in $\AC_0$ for this distribution, without computing the majority function along the way.

\subsubsection{Depth-4 lower-bounds}\label{sec:GRSS}

G\"o\"os, Riazanov, Sofronova, and Sokolov \cite{goos2023top} proved a top-down lower-bound for depth-4 circuits computing Parity. The extra layer is treated using the \emph{spreadness lemma} \cite{Rao2020Coding}, and after this the argument is a non-trivial  variant of Håstad, Jukna and Pudl{\'a}k's depth-3 lower-bound. At a crucial juncture in the proof of Claim 4 (Lemma 4, page 7 of ArXiv version v2), they use the property that every input $x$ has sensitivity $S(f,x) = n$. This is the only time they use the fact that $f$ is parity, but of course, this property is the defining property of the parity function (or its negation). The GRSS proof shows that parity has neither $\Sigma_4$ nor $\Pi_4$ circuits. 
To begin, let us naturalize this proof to show that a random function from a biased distribution does not have $\Sigma_4$ circuits. By naturalizing this proof we will extract a distinguisher which distinguishes all easy functions from a random function of the appropriate bias. In the same way that we argued earlier, $\AC_0$-computable PRFs $F$ with bias $\Pr[f(x) = 1]=1-2^{-q}$ can be constructed from any unbiased PRF $F'$, by adding only one extra layer of OR gates on top. So the $\AC_0$-natural-proofs barrier also applies to biased functions.

Let $f$ denote a random function, where each output $f(x)$ is chosen independently with bias $\Pr[f(x) = 1] = 1 - \frac{1}{n^{1-\delta}}$ where $\delta=0.1$. Then, with high probability over the choice of $f$, it will happen for every $x \in f^{-1}(0)$, that the \emph{$2$-wise sensitivity} is high: 
\[
\forall x \in f^{-1}(0) \qquad\qquad
S_2(f,x) = \frac{|\{ z\in\binom{[n]}{2} \mid f(x \oplus z)\neq f(x)\}|}{\binom{n}{2}} \ge 1 - \frac{2}{n^{1-\delta}}.
\]
To see this, note for any fixed $x$ the expected number $Z$ of $2$-neighbours $x'$ with $f(x') = 0$ is precisely $\mu =\frac{1}{n^{1-\delta}} \binom{n}{2}\gg 6n$, and by a Chernoff bound,
\[
\Pr\left[S_2(f,x) < 1 - 2 \frac{1}{n^{1-\delta}}\right] = \Pr[Z > 2 \mu] \le \exp\left(-\frac{1}{3} \mu\right) = o( 2^{-2n} )
\]
and hence w.h.p. there will be no such $x \in f^{-1}(0)$. This means that, for every $x \in f^{-1}(0)$, there are fewer than $\frac{2}{n^{1-\delta}} \binom{n}{2}$ strings in $f^{-1}(0)$ at Hamming distance $2$ from $x$. Our natural property will be a small variation of this. We actually want the sensitivity to be high even if I have chosen some set $I$ of size $\le \frac{n}{10}$ (say), and am only considering $2$-neighbors which do not change any coordinate in $I$. Our $\AC_0$ natural property is then precisely the following:
\[
\forall x \in f^{-1}(0) \forall I \in \binom{[n]}{\le \frac{n}{10}} \qquad \qquad S_{2,I}(f,x) = \frac{|\{ z\in\binom{[n]\setminus I}{2} \mid f(x \oplus z)\neq f(x)\}|}{\binom{n-|I|}{2}} \ge 1 - \frac{2}{n^{1-\delta}}
\]
while $f$ also satisfies $|f^{-1}(0)| \geq 2^{n-o(n)}$. The same proof as above shows that it is large among random functions chosen with the bias $1 - \frac{1}{n^{1 - \delta}}$. It is in $\AC_0$ by just doing a large AND over all $x$ and $I$ and checking the sensitivity requirement, and doing an approximate counting of the $0$s in the truth table.

Now let us turn our attention to the GRSS proof. We assume familiarity with the proof and only detail what changes need to be made to the proof in order for it to work for any function with the above property. The sensitivity of parity is used only once, crucially, in the proof of Claim 4. At this point in the proof we found a maximal set $I$ violating the spreadness condition. If the indicator $x''$ of $I$ is in $f^{-1}(1)$, then the proof is done. Otherwise, the proof takes $i_0 \in [n]\setminus I$ having the highest probability of being $1$ among the elements of $B' = \{x \in A : x_I = 1_I\} - x''$. One argues that this probability is at least $\frac{m-I}{n-I}$, by noticing that $|\{(z,i_0) \mid z \in B', z_{i_0}=1\}| = |B'| (m - |I|)$, hence by the pigeonhole principle some $i_0 \in [n]\setminus I$ must appear in at least $\frac{1}{n-|I|}$ such pairs, and hence the probability that $z \in B'$ has $z_{i_0} = 1$ is at least $\frac{m-|I|}{n-|I|}$.

To naturalize this part of the argument, we need something slightly stronger, we need to show that there are many choices $\{i_0, i_1\}$ of \textit{two} coordinates that have reasonably high probability of being both set to $1$. Indeed, we can proceed greedily. Repeating the same argument we just did, now with a set of triples, we have $|\{(z,i_0,i_1) \mid z \in B', z_{i_0 i_1}=1^2\}| = |B'| \binom{m - |I|}{2}$, and hence there exists some choice $\{i^{(1)}_0, i^{(1)}_1 \}$ with $\Pr_{z \in B'}[z_{i^{(1)}_0 i^{(1)}_1} = 1^2] \ge \frac{\binom{m-|I|}{2}}{\binom{n-|I|}{2}}$. Now recall that $B'$ is $(1/p)^{(1-\eps)}$-spread, and hence $\Pr_{z \in B'}[z_{i^{(1)}_0 i^{(1)}_1} = 1^2] \le p^{2}$. So we have only removed $p^{2}$ mass out of our set of triples. This can be done $\frac{2}{p^{2}} \gg n^{1+\delta}$ times before half the mass has been removed. We then conclude that there are $\gg n^{1+\delta}$ ``high-probability pairs'' $\{i_0,i_1\}$ such that $\Pr_{z \in B'}[z_{i_0 i_1} = 1^2] \ge \frac{\binom{m-|I|}{2}}{2\binom{n-|I|}{2}}$. Since $x''$ has $2$-sensitivity $\ge 1 - \frac{2}{n^{1-\delta}}$, the number of distance-2 neighbors $x'$ of $x''$ with $f(x') = 1$ is at least $(1-\frac{2}{n^{1-\delta}})\binom{n}{2} = \binom{n}{2} - O(n^{1+\delta})$. It follows that one of these distance-2 neighbors is equal to $x'' \oplus 1_{i_0 i_1}$ for one of our ``high-probability'' pairs. We then set $I' = I \cup \{i_0, i_1\}$, and then the calculation proceeds as in the proof of Claim 4. The detailed changes needed are: (1) we only get $(1/p)^{1 - 3.1\eps}$-spreadness for $B$, so $2.1$ gets replaced with $3.1$; (2) the probability $\Pr_{x\sim A}[x_{I'}=1_{I'} \mid x_I = 1_I]$ is now at least $\frac{\binom{m-|I|}{2}}{2\binom{n-|I|}{2}} \ge p^2 \frac{1}{4}$; (3) so now, analogous to their proof, the $p^{-2\eps |J|}$ factor defeats the $p^{2\eps}$, and the $p^{-0.1\eps}$ factor defeats the $\frac{1}{4}$.

\medskip
We have then concluded that the GRSS proof is natural in the above sense (where largeness is with respect to biased functions), and that any barrier established for unbiased functions will also give us a barrier here. To finish, let us just remark that our proof shows that our biased function does not have $\Sigma_4$ circuits, but the original proof shows a simultaneous lower-bound for $\Sigma_4$ and $\Pi_4$. To correct for this, we can change our property so that we require large sensitivity with respect to $0$ among inputs whose first bit is $0$, and large sensitivity with respect to $1$ for inputs whose first bit is $1$. A peculiar feature of our depth-4 $\AC_0$-natural lower bound, which was not the case for the other naturalized variants we saw earlier, is that it does not show a lower bound for the originally studied function (in this case Parity).

\section{A weak barrier based on bounded independence}\label{sec:k-wise-independence}

Braverman's celebrated result states that polylogarithmic independent distributions fool $\AC_0$ circuits \cite{Braverman}. Here, we consider an improvement of it due to Tal \cite{tal2017tight}.

\begin{theorem}[Tal~\cite{tal2017tight}]\label{Braverman}
Every $k$-wise independent distribution over $\{0,1\}^n$ with $k \geq O(\log (S_\Phi/\epsilon))^{3d_\Phi+3}$, $\epsilon$-fools every distinguisher $\Phi: \{0,1\}^n \rightarrow\{0,1\}$ computed by size $\le S_\Phi$ circuits of depth $\le d_\Phi$.

In the case of $d_\Phi=2$, $k \geq O(\log S_\Phi \log (S_\Phi/\epsilon))$ suffices.
\end{theorem}

We will use polynomials over finite fields to generate such $k$-wise independent distributions locally.

\begin{restatable}{theorem}{altevalPoly} \label{evalPoly}
    Let $n, k, c_1, c_2 \in \mathbb{N}$. There exists a circuit of depth $c_1+c_2+2$ and size $2^{\tilde{O}(kn^2)^{1/c_1}+n^{2/c_2}}$ that on input $(\bar{a},y) \in (\ZO^n)^k \times \ZO^n$ outputs $p(y) \in \ZO^n$ where $p = \sum_{i=0}^{k-1} a_i z^i \in \F[z]$ and $y \in \F$.

    All the output gates of the constructed circuit are $\mathsf{AND}$ gates.
\end{restatable}

The circuit construction uses ideas from \cite{HealyViola} to implement multiplication in $\F$. Their construction was assuming polynomial size and unbounded arity parity gates. We modify their construction for subexponential size with no parity gates. The proof can be found in Appendix \ref{sec:deferredproofs}.

The above gives us a natural approach to create a PRF, using the fact that the outputs of a degree $k-1$ polynomial chosen uniformly at random exhibit $k$-wise independence. For $x \in \F$ let $\mathrm{LSB}(x)$ be the least significant bit of the $n$-bit representation of $x$.

\begin{definition}[Distribution $\mu_{n,k}$]
    For $n \in \mathbb{N}$  and $k \leq 2^{n}$ define the distribution $\mu_{n,k}$ generated by the following process: Sample uniformly at random $\bar{a} \in \mathbb{F}^k_{2^{n}}$ and consider the polynomial $p(x)=\sum_{i=0}^{k-1} a_i x^i$. Output the truth table of $f(x)=\mathrm{LSB}(p(x))$.
\end{definition}

\begin{claim} \label{claim:mukwise}
    The distribution $\mu_{n,k}$ defined above is a $k$-wise independent distribution over $\{0,1\}^{2^{n}}$.
\end{claim}
\begin{proof}

    Let $F = \{c_1,\dots,c_k\} \subseteq \mathbb{F}_{2^{n}}$ be of size exactly $k$.
    The linear map from coefficients to evaluations in $F$, $(a_1,\dots,a_k) \mapsto (p(c_1), \dots, p(c_k) )$, is a bijection because the map can be inverted uniquely via polynomial interpolation.
    Since the coefficients of the polynomial, $(a_1,\dots,a_k)$, are uniformly distributed over $\mathbb{F}^k_{2^{n}}$, we must have that $(p(c_1), \dots, p(c_k))$ are uniformly distributed over $\mathbb{F}^k_{2^{n}}$. And their representations as $kn$ bit strings is hence also uniformly distributed. Since their least significant bits are a marginal, they are also uniform. Hence every $k$-sized marginal of $\mu_{n,k}$ is uniformly distributed over $\{0,1\}^k$.    
\end{proof}

\begin{claim}\label{claim:muDSFG}
    For any $c_1,c_2 \in \mathbb{N}$ the distribution $\mu_{n,k}$, defined as above, is an $(n, d, S)$-function generator with $d=c_1+c_2+2$ and $S=2^{\tilde{O}(kn^2)^{1/c_1}+n^{2/c_2}}$.
\end{claim}
\begin{proof}
    For any  $n,k$ the distribution $\mu_{n,k}$  is supported on the truth tables of circuits that take an input in $\mathbb{F}_{2^{n}}$, evaluate a degree $k-1$ polynomial on that input, and extract the last bit. The claim directly follows from \Cref{evalPoly}.
\end{proof}

We can now put it all together and get our constant-depth PRF.

\begin{theorem}\label{thm:muPRF}
    For any $c_1,c_2 \in \mathbb{N}$ there exists $d=c_1+c_2+2$, $S=2^{\tilde{O}(\log(S_\Phi/\epsilon)^{3d_\Phi+3}n^2)^{1/c_1}+n^{2/c_2}}$ such that the distribution $\mu_{n,k}$, for an appropriate $k$, is a $(n,d,S,d_\Phi,S_\Phi,\epsilon)-$pseudorandom function generator.

    When $d_{\Phi}=2$, the bound on $S$ improves to $2^{\tilde{O}(\log(S_\Phi)\log(S_\Phi/\epsilon)n^2)^{1/c_1}+n^{2/c_2}}$.
\end{theorem}
\begin{proof}
By \Cref{claim:mukwise} we know that $\mu_{n,k}$ is $k-$wise independent.
\Cref{Braverman} gives us bounds on $k$ such that the PRF $\mu_{n,k}$ $\epsilon$-fools every distinguisher computed by size $\le S_\Phi$ circuits of depth $\le d_\Phi$. Substituting this value of $k$ in \Cref{claim:muDSFG} gives a circuit of the required size and depth computing $\mu_{n,k}$.
\end{proof}

\Cref{lem:prfcorollary} directly tells us that for any $c \in \mathbb{N}$ there is no $(n,d,S,d_\Phi,S_\Phi,\epsilon)-$natural proof where $d,S$ are as above. We now see what parameters $d_\Phi,S_\Phi$ are needed to get a $2^{n^\eps}$ lower bound for depth $d$ circuits.

\begin{corollary}\label{thm:weak-unconditional-barrier}
    Let $d \in \mathbb{N},d\geq 5$ and $\eps \geq 2/(d-4)$ be a fixed constant. Let $\Phi = \{\Phi_n : \Q^{2^n} \rightarrow \Q\}$ be a family of $(n, d, 2^{n^\eps}, d_{\Phi}, S_{\Phi}, 1/2)$-natural properties. If $S_\Phi \leq 2^{n^\alpha}$, then $\eps \leq (3\alpha(d_\Phi+1)+4)/(d-3)$.
\end{corollary}

\begin{proof}
    Set $c_2 = \lceil2/\eps\rceil \leq 2/\eps+1$. Now $c_1 = d-c_2-2\geq d-2/\eps-3$. We know such a choice is possible because $\eps \geq 2/(d-4)$.
    
    If $S_\Phi \leq 2^{n^\alpha}$, then the maximum lower bound provable is at most $2^{\tilde{O}(n^{(\alpha(3d_\Phi+3)+2)/c_1})+n^{2/c_2}}$. For $2^{n^\eps}$ to be at most this bound we need $2/c_2 \geq \eps$ or $(\alpha(3d_\Phi+3)+2)/c_1 \geq \eps$. Since we ensured the former does not hold it must be the latter that holds. Simplifying it we get $\eps \leq (3\alpha(d_\Phi+1)+4)/(d-3)$.
\end{proof}

In other words the best depth-$d$ lower bound you can hope to prove with a natural property of depth $d_\Phi$ and size $2^{n^\alpha}$ is of the form $2^{n^{O(\alpha d_\Phi/d)}}$. Hence in order to prove natural lower bounds for large subexponential sizes (i.e. $2^{n^{\beta/d}}$ for a large beta) it is necessary that the depth of the property $d_\Phi$ increases to a larger constant or the parameter $\alpha$ in the size becomes larger. In the next section we show a stronger statement: even increasing to larger constant depths will not help.

\section{A stronger barrier based on a derandomized Switching Lemma}\label{sec:unconditional}

We will now show PRFs of constant depth $d$ and subexponential size (of the form $2^{n^{O(1/d)}}$) that fool all of $\AC_0$. It follows that no $\AC_0$-natural property can get depth-$d$ size lower bounds better than $2^{n^{\Omega(1/d)}}$. In fact, this is true even for properties computable by super-constant depth $(\log n)^{\Omega(1)}$ and quasipolynoimal size $2^{n^{O(1)}}$! See \Cref{thm:unconditional-natural-proofs} for a formal statement with the precise bounds.

At a high level our PRF is built in a similar manner to the Ajtai-Wigderson~\cite{ajtai1985deterministic} PRG that fools $\AC_0$ circuits. Our construction uses two main ingredients: a PRG that fools small-depth decision trees, and a distribution on pseudorandom restrictions that collapses constant depth circuits to small-depth decision trees. These ingredients are combined using the ``Ajtai-Wigderson framework'' as described in the survey by Hatami and Hoza~\cite{HatamiH24}. We state the framework adapted to our use case below.

In the rest of this section we will be using $N$ to represent both $2^n$ and we will overload notation to also let $[N]$ represent $\Q^n$. The notation $[u \in W]$ represents a bit that is $1$ if and only if $u \in W$.

\begin{lemma}\label{lemma:aw}[The Ajtai-Wigderson Framework, adapted]
    Let $\calC$ be a class of circuits closed under restrictions and $\calC'$ be a class of circuits perfectly fooled by a PRG $G'$.
    Let $\mathcal{W}$ be a distribution on subsets of $[N]$ such that
    \begin{itemize}
        \item for all $u\in [N]$, $\Pr_{W \sim \mathcal{W}}[u \in W] \leq 1-p$, and
        \item for all $C:\Q^N \to \Q \in \calC$, $\Pr_{W \sim \mathcal{W},X \sim \ZO^N}[C|_{W \mapsto X} \in \calC'] \geq 1-\epsilon^*$.
    \end{itemize}
    Let $\ell = 1/p \cdot \ln (N/\epsilon^*)$. Then there is a PRG $G$ that $(\ell+1)\epsilon^*$-fools $\calC$.

    Furthermore if every string $X$ generated by $G'$ has a depth-$d$ size-$s$ circuit computing $u \mapsto X_u$, and every set $W$ generated by $\mathcal{W}$ has a depth-$d$ size-$s$ circuit computing $u \mapsto [u \in W]$, then for every output $Z$ of $G$ there is a depth-$(d+2)$ size-$O(s\ell)$ circuit computing $u \mapsto Z_u$.
\end{lemma}

\begin{proof}
    The PRG $G$ is constructed as follows.
    Sample $W^{(1)},\dots,W^{(\ell)}$ independently from $\mathcal{W}$ and $X^{(1)},\dots,X^{(\ell)}$ independently from $G'$. Let $\rho_i$ denote the restriction $\overline{W^{(i)}} \mapsto X^{(i)}$ and $\rho_{\ell+1} := [N] \mapsto 0^N$. Output $Z \in \ZO^N$ defined by the restriction $\rho_1\rho_2\cdots\rho_{\ell}\rho_{\ell+1}$.

    The correctness of the PRG is covered in the survey~\cite[Theorem 5.21]{HatamiH24}. The circuit construction for computing $u \mapsto Z_u$, assuming we have computed each $X^{(i)}_u$ and each $[u \in W^{(i)}]$, is shown in~\Cref{fig:AW}. Since $Z_u = X^{(i^*)}_u$ for the first index $i^* \in [\ell+1]$ such that $u \not\in W^{(i^*)}$, the circuit computes $X^{(i)} \wedge (\wedge_{j < i} [u \in W^{(j)}]) \wedge [u \not\in W^{(i)}]$. Finally an $\mathsf{OR}$ is taken over all $i$. The negations do not add to the depth of the circuit since they can be pushed to the inputs with only a blowup of $2$ in the overall size of the circuit.
\end{proof}

\tikzstyle{and}=[fill=white, draw=black, shape=circle]
\tikzstyle{generator}=[fill=white, draw=black, shape=rectangle,minimum width=0.8cm,minimum height=0.8cm,inner sep=5pt]

\tikzstyle{new edge style 0}=[-, line width=2.5pt]
\tikzstyle{new edge style 1}=[-, dashed]
\tikzstyle{dotstyle}=[-, dotted]

\begin{figure}
    \centering
\begin{tikzpicture}
	\begin{pgfonlayer}{nodelayer}
		\node [style=none] (0) at (-9, -3) {};
		\node [style=none] (1) at (-7.3, -3) {};
		\node [style=none] (2) at (-5.475, -3) {};
		\node [style=none] (3) at (-3.75, -3) {};
		\node [style=none] (4) at (-2, -3) {};
		\node [style=none] (5) at (-0.025, -3) {};
		\node [style=none] (6) at (-7.3, -2) {};
		\node [style=none] (7) at (-3.75, -2) {};
		\node [style=and] (9) at (-8.5, 0.25) {$\wedge$};
		\node [style=and] (10) at (-5, 0.25) {$\wedge$};
		\node [style=and] (11) at (-1.7, 0.25) {$\wedge$};
		\node [style=and] (12) at (-5, 1.5) {$\vee$};
		\node [style=and] (13) at (-7.25, -0.75) {$\neg$};
		\node [style=and] (14) at (-2.725, -1) {$\neg$};
		\node [style=none] (19) at (-5, 2.5) {};
		\node [style=none] (20) at (-9, -3.5) {$X^{(1)}_u$};
		\node [style=none] (23) at (-7.3, -3.5) {$[u \notin W^{(1)}]$};
		\node [style=none] (25) at (-5.475, -3.5) {$X^{(2)}_u$};
		\node [style=none] (27) at (-3.75, -3.5) {$[u \notin W^{(2)}]$};
		\node [style=none] (28) at (-2, -3.5) {$X^{(3)}_u$};
		\node [style=none] (29) at (0.025, -3.5) {$[u \notin W^{(3)}]$};
		\node [style=none] (33) at (-6.25, 1.5) {};
		\node [style=none] (34) at (-2.7, 1.5) {};
		\node [style=none] (37) at (-0.225, 0.25) {$\dots$};
		\node [style=none] (39) at (0.75, -2) {$\dots$};
		\node [style=none] (42) at (-2.375, 1.35) {$\ell$};
		\node [style=none] (43) at (1.5, -3.5) {$\dots$};
		\node [style=none] (44) at (-4.5, 2) {};
		\node [style=none] (45) at (-4.5, 2.25) {$Z(u)$};
	\end{pgfonlayer}
	\begin{pgfonlayer}{edgelayer}
		\draw [in=270, out=90] (1.center) to (6.center);
		\draw [in=270, out=90] (3.center) to (7.center);
		\draw (0.center) to (9);
		\draw (2.center) to (10);
		\draw (4.center) to (11);
		\draw (9) to (12);
		\draw (10) to (12);
		\draw (11) to (12);
		\draw (14) to (11);
		\draw (12) to (19.center);
		\draw [style=new edge style 1, in=-150, out=-30, looseness=1.25] (33.center) to (34.center);
		\draw (13) to (10);
		\draw (13) to (11);
		\draw (10) to (7.center);
		\draw (7.center) to (14);
		\draw (9) to (6.center);
		\draw (6.center) to (13);
		\draw (11) to (5.center);
	\end{pgfonlayer}
\end{tikzpicture}

\caption{The circuit construction for the PRG from \Cref{lemma:aw}.}\label{fig:AW}
\end{figure}
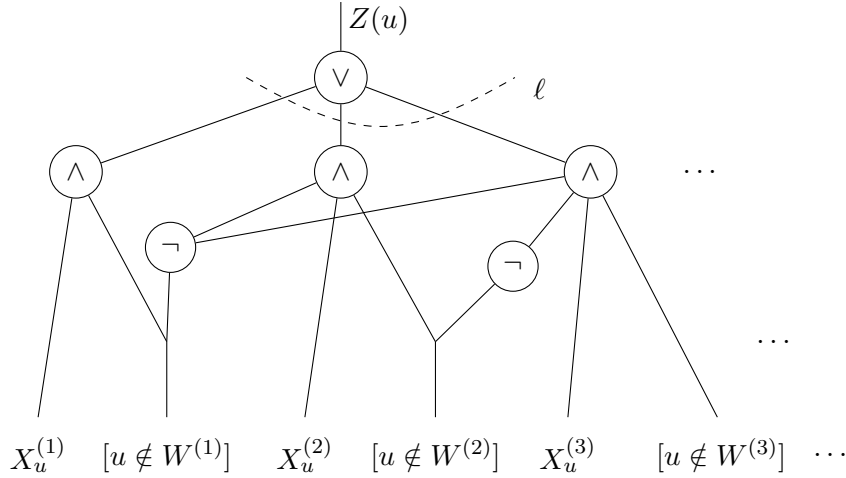

Our goal is to create a PRF that $\epsilon$-fools depth-$d_{\Phi}$ size-$S_{\Phi}$ circuits, and so this is our class $\calC$. The class $\calC'$ will be decision trees of depth $4 \log S_{\Phi}$, which are perfectly fooled by any $(4 \log S_{\Phi})$-wise independent PRF $G'$. All that is left for us to do is to provide the distribution $\mathcal{W}$ that simplifies $\calC$ to $\calC'$ as described in \Cref{lemma:aw}, and to show that both $G'$ and $\mathcal{W}$ have efficient constant depth circuits, which we will do in \Cref{subsec:prf-construction}.

For the distribution $\mathcal{W}$ we will use the derandomized switching lemma by Trevisan and Xue~\cite{trevisan2013derandomized}. We describe this lemma below, but first we need a tool to generate biased random restrictions from unbiased random strings.

\begin{definition}
    Let $q \in \N$, and let $\tau : \ZO^q \times \ZO \to \{0,1,*\}$ be defined as $$\tau : (x,b) \mapsto \begin{cases}
        * \text{ if }\wedge_{i \in [q]} x_i \\
        b \text{ otherwise.}
    \end{cases}$$
    Then $D_{q,N}^{\uparrow} : \ZO^{(q+1)N} \to \{0,1,*\}^N$ is defined as $(y^{(1)},\dots,y^{(N)}) \mapsto (\tau(y^{(1)}),\dots,\tau(y^{(N)}))$.
\end{definition}

\begin{theorem}[Main Lemma (Lemma 7) in \cite{trevisan2013derandomized}]\label{altthm:derandomizeSL}
    Let $C$ be a $M$ clause $t$-CNF on $N$ bits, $s>0$, $q \in \N$ be a positive parameter and $\mathsf{E}$ a distribution over $\{0,1\}^{(q+1)N}$ that $\epsilon$-fools all $M2^{t(q+1)}$-clause CNFs. Then
    
    $$\Pr_{r \sim \mathsf{E}} [\mathsf{DTdepth}(C|_{D^{\uparrow}_{q,N}(r)})>s] \leq 2^{s+t+1} (5t/2^q)^s + \epsilon \cdot2^{(s+1)(2t+\log M)}.$$
\end{theorem}

In their PRG, Trevisan and Xue use the above lemma, with the distribution $E$ instantiated with a small-bias PRG to collapse $\AC_0$ circuits to small depth decision trees. We will do the same but with the distribution $E$ instantiated from~\Cref{thm:muPRF} in~\Cref{sec:k-wise-independence}, since that gives us constant depth PRFs fooling depth 2 circuits. The collapsing of the $\AC_0$ circuit follows the same steps as when the usual switching lemma is used: A random restriction is used to make the bottom fan-in small. Then $d-1$ random restrictions are used to decrease the depth of the $\AC_0$ circuit from $d$ to $2$, with the last restriction changing the depth-$2$ circuit to a small depth decision tree.

The initial random restriction and the subsequent ones are sampled from the two distributions below respectively.

\begin{definition}[Restriction-Generating Distributions]\label{def:restgen}
    Let $\epsilon'_0$ and $\epsilon'_1$ be parameters to be set in \Cref{clm:pruning} and \Cref{clm:switching} respectively. Distributions $\sfD_0$ and $\sfD_1$ are defined below.
    \begin{itemize}
        \item A random restriction from the distribution $\sfD_0$ is sampled as follows. A string $r$ of length $(6+1)N$ is sampled from a distribution $\sfE_0$ that $\epsilon'_0$-fools $12S_\Phi$-size CNFs. The last bit in each of the $N$ blocks is then resampled from a uniform distribution. The restriction is then $D_{6,N}^{\uparrow}(r)$. We will also require that $\sfE_0$ is $6$-wise independent.
        \item Let $q$ and $t$ be parameters to be set later (in \Cref{clm:switching}). A random restriction from the distribution $\sfD_1$ is sampled as follows. A string $r$ of length $(q+1)N$ is sampled from a distribution $\sfE_1$ that $\epsilon'_1$-fools $S_\Phi2^{t(q+1)}$-size CNFs. The last bit in each of the $N$ blocks is then resampled from a uniform distribution. The restriction is then $D_{q,N}^{\uparrow}(r)$.  We will also require that $\sfE_1$ is $q$-wise independent.
    \end{itemize}
    For simplifying some computations we will also make use of the fact that when we instantiate $\sfE_0$ and $\sfE_1$ it will be with distributions that are also $6$-wise and $q$-wise independent respectively (see the proof of \Cref{lemma:constructingW}).
\end{definition}

\begin{remark}
    \begin{itemize}
        \item Since in the definition above $\sfE_0$ and $\sfE_1$ are at least $6$-wise and $q$-wise independent respectively, this guarantees that for each index $i$, $\rho(i)=\ast$ with probability exactly $1/64$ if sampled from $\sfD_0$ and $1/2^q$ if sampled from $\sfD_1$.
        \item In the above definition the random strings $r$ have some bits resampled at uniform. They will continue to fool CNFs after these resamplings (see \Cref{prop:foolingmanipulations}) and can hence be used in \Cref{altthm:derandomizeSL}.
        \item $\sfE_0$ and $\sfE_1$ are instantiated only in the proof of \Cref{lemma:constructingW}, and their $k$-wise independence can be verified in that proof.
    \end{itemize}
\end{remark}

With these distributions defined we can now state how we use the Ajtai-Wigderson framework.

    \begin{lemma}\label{lemma:awapplied}[Applying the framework]
        The conditions of \Cref{lemma:aw} are satisfied with:
        \begin{itemize}
            \item $\epsilon^* = 8/S_{\Phi}$,
            \item $\calC$ being the class of depth-$d_{\Phi}$ size-$S_{\Phi}$ circuits,
            \item $\calC'$ being the class of depth-$4\log(S_\Phi)$ decision trees and $G'$ being a $4\log(S_\Phi)$-wise independent distribution,
            \item $p=1/(2^{O(d_{\Phi})}(\log S_{\Phi})^{d_{\Phi}-1})$, and
            \item $\mathcal{W}$ being the distribution of the set of restricted indices of $\rho := \rho_0\cdots\rho_{d_\Phi-1}$ where $\rho_0 \sim \sfD_0$ and $\rho_i \sim \sfD_1$ for $i>0$, with $q=\log \log S_{\Phi}+O(1)$, $t=4\log(S_{\Phi})$ and $\epsilon'_1 = (2/S_\Phi^2) \cdot 2^{-60\log^2 S_\Phi}$ in the definition of $\sfD_1$.
        \end{itemize}
    \end{lemma}

    Before we prove this we first show the purpose of the two distributions individually. First we verify that the bottom fan-ins are reduced by the first restriction. It is easy to see that this happens with a truly random restriction, the following claim shows that a pseudorandom restriction that fools CNFs would also suffice.

    \begin{claim}\label{clm:pruning}
        Let $S_{\Phi} \in \N, \epsilon^*>0$ and set $\epsilon'_0 := \epsilon^*/4S_{\Phi}$ in the definition of $\sfD_0$ (\Cref{def:restgen}). Let $C$ be any circuit of size at most $S_{\Phi}$. The random restriction $\rho_0 \sim \sfD_0$ will, with probability at least $1-\epsilon^*/2$, trivialize every bottom gate of $C$ with fan-in larger than $1.1\log (4S_\Phi/\epsilon^*)$.
    \end{claim}  
    \begin{proof}
        Recall that $\rho_0$ is sampled by sampling a string $r$ from a CNF-fooling distribution $E$ and applying $D^{\uparrow}_{7,N}$ to it. Our claim follows by the observation that there is a small DNF that implies that the restriction obtained from $r$ trivializes a gate, and that the DNF itself has a high probability of being satisfied by a uniform $r$.
        
        Let $g$ be an AND/OR gate with $t$ input bits $x_{i_1},\dots,x_{i_t}$. The gate $g$ gets trivialized if $\bigvee_{j \in [t]} (\rho_0(i_j) = b_j)$, where $b_j$ is the value of the bit $x_{i_j}$ that would trivialize $g$. The statement $\rho_0(i_j) = b_j$ can be written as $\bigvee_{a \in [6]} (r^{(i_j)}_a = 0 \wedge r^{(i_j)}_7 = b_j)$ where $r^{(i_j)}$ is the $i_j$th block of $r$. Hence there is a DNF of size $12t$ that implies that $g$ is trivialized.

        Under a uniform $r$, the probability that any particular input bit to $g$ is set to the trivializing $b_j$ is $\frac{1}{2}(1 - 1/2^6)$, so the probability that the DNF is not satisfied is $(1/2 + 1/2^7)^t < (1/2)^{t/1.1}$.

        Since the actual distribution of $r$ (i.e. $\sfE_0$) $\epsilon'_0$-fools this DNF the probability that this DNF is not satisfied is at most $(1/2)^{t/1.1} + \epsilon'_0$. By a union bound over the at most $S_\Phi$ bottom gates and substituting the value of $\epsilon'_0$ the probability that there is a gate with fan-in more than $1.1 \log (4S_\Phi/\epsilon^*)$ that does not trivialize is at most $S_\Phi((1/2)^{\log (4S_\Phi/\epsilon^*)} + \epsilon^*/4S_\Phi) = \epsilon^*/2$.
    \end{proof}

    Now we verify that the subsequent restrictions simplify the circuit. This crucially uses \cref{altthm:derandomizeSL}.

    \begin{claim}\label{clm:switching}
        Let $\epsilon^* \geq 8/S_{\Phi}$, and set $\epsilon'_1 := (\epsilon^*/4S_{\Phi}) \cdot 2^{-15\log^2(8S_{\Phi}/\epsilon^*)}$, $t=2\log(8S_\Phi/\epsilon^*)$ and $q=\lceil\log(120\log S_{\Phi})\rceil$ in the definition of $\sfD_1$ (\Cref{def:restgen}).
        
        Let $C$ be any depth $d_\Phi$ circuit where the bottom gates have fan-in at most $2\log(8S_{\Phi}/\epsilon^*)$ and with at most $S_{\Phi}$ non-bottom gates. The random restriction $\rho \sim \sfD_1$ will, with probability at least $1-\epsilon^*/2S_{\Phi}$, simplify every gate in $C$ at depth $d_\Phi-1$ to a decision tree of depth at most $2\log(8S_\Phi/\epsilon^*)$.

        Consequently the restricted circuit $C$ can be written as a depth $d_\Phi-1$ circuit where the depth-$(d_\Phi-1)$ gates have fan-in at most $2\log(8S_{\Phi}/\epsilon^*)$ and with at most $S_{\Phi}$ non-bottom gates.
    \end{claim}
    \begin{proof}
        The first part is a simple application of \Cref{altthm:derandomizeSL} with $s=t=2\log(8S_\Phi/\epsilon^*)$ and $q=\lceil\log(120\log S_{\Phi})\rceil$ and a union bound on the at most $S_{\Phi}$ gates at depth $d_\Phi-1$. These parameters do require $\epsilon^* \geq 8/S_{\Phi}$ in order to work.

        The second part is a routine application of the switching lemma: Gates at depth $d_\Phi-2$ are now an AND/OR of $2\log(8S_{\Phi}/\epsilon^*)$-depth decision trees. These decision trees can be written both as a CNF and a DNF of width $2\log(8S_{\Phi}/\epsilon^*)$. Hence the gates at depth $d_\Phi-2$ can be written as either a CNF or a DNF of width $2\log(8S_{\Phi}/\epsilon^*)$.
    \end{proof}

    \begin{proof}[Proof of \Cref{lemma:awapplied}]
        That $G'$ perfectly fools $\calC'$ is folklore. For a $C \in \calC$ using \Cref{clm:pruning} and \Cref{clm:switching} we see that the restriction $\rho$ will with probability $\geq 1-\epsilon^*$ simplify $C$ to a decision tree of depth at most $2\log(8S_\Phi/\epsilon^*)\leq 4\log(S_\Phi)$. Since the definition of these restrictions have the bit values resampled to be truly random and independent of the locations, the locations to be restricted satisfy the condition in~\Cref{lemma:aw}. Let $W$ be the indices of bits restricted by $\rho$. In order to have $u \not\in W$, it must not have been restricted by any of $\rho_0$ to $\rho_{d-1}$, which would happen with probability $p=2^{-q(d_{\Phi}-1)}/64 \leq 2^{O(d_{\Phi})}(\log S_{\Phi})^{d_{\Phi}-1}$.
    \end{proof}

    \Cref{lemma:aw} now gives us the PRG that we need. We now analyze its locality as a constant depth PRF.

\subsection{Construction of the PRF}\label{subsec:prf-construction}

Since the distinguishers we are interested in have size between $2^n$ and $2^{2^n}$, we will always assume $\log S_\Phi \geq n$ and $\log \log S_\Phi \leq n$.

\Cref{lemma:aw,lemma:awapplied} reduce the construction of our main PRF to the construction of the PRF $\mathcal{W}$ and of a PRF computing a $4\log(S_\Phi)$-wise independent distribution $G'$. The latter is a direct corollary of \Cref{claim:mukwise} and \Cref{claim:muDSFG}.

\begin{corollary}
    For every $c_1,c_2 \in \mathbb{N}$ there is a $(n,c_1+c_2+2,2^{\tilde{O}(\log(S_\Phi) n^2)^{1/c_1}+n^{2/c_2}})$-function generator whose output is $4\log S_\Phi$-wise independent.
\end{corollary}

\begin{lemma}\label{lemma:constructingW}
    For every $c_1,c_2 \in \mathbb{N}$ and every $W$ in the support of $\mathcal{W}$ there are circuits of depth $c_1+c_2+2$ and size $d_\Phi2^{\tilde{O}(\log^3(S_\Phi)n^2)^{1/c_1}+n^{2/c_2}}$ that compute $u \mapsto [u \in W]$.
\end{lemma}
\begin{proof}
    From \Cref{lemma:awapplied} we see that $W$ is built by composing together a restriction from $\sfD_0$ along with $d_\Phi-1$ restrictions from $\sfD_1$ to get a restriction $\rho$, and taking the set of indices that are restricted in $\rho$.

    $\sfD_0$ and $\sfD_1$ in turn are built from $\sfE_0$ and $\sfE_1$ (see \Cref{def:restgen}). Since $\sfE_1$ fools larger CNFs to a larger extent, we will replace $\sfE_0$ with a marginal of $\sfE_1$ and so we will focus only on circuits implementing $\sfE_1$. $\sfE_1$ outputs strings of length $(q+1)N$ where $q=\log \log S_\Phi+O(1) \leq n + O(1)$. We will view it as a PRF on $n + \lceil \log (q+1) \rceil = n(1+o(1))$ input bits so that it outputs at least $(q+1)N$ bits. Recall (from \Cref{def:restgen} and \Cref{lemma:awapplied}) that $\sfE_1$ must fool $S_\Phi2^{t(q+1)}$-size CNFs to within error $\epsilon'_1 = (2/S_{\Phi}^2) \cdot 2^{-60\log^2 S_{\Phi}} \leq 2^{-O(\log^2 S_{\Phi})}$. Substituting the values of $t$ and $q$ it needs to fool $2^{O(\log S_\Phi \log \log S_\Phi)}$-size CNFs.

    We use \Cref{thm:muPRF} with the following parameters (since that uses the same named parameters that our current lemma uses, we will refer to its parameters as $n',c_1',c_2',d_\Phi',S_\Phi'$ and $\epsilon'$): $n' = n + \lceil \log (q+1) \rceil$, $c_1' = c_1$, $c_2'=c_2$, $d_\Phi' = 2$, $S_\Phi' = 2^{O(\log S_\Phi \log \log S_\Phi)}$ and $\epsilon' = 2^{-O(\log^2(S_{\Phi}))}$. It gives us the PRF $\sfE_1$ supported on truth tables of circuits of depth $d=c_1+c_2+2$ and size $S=2^{\tilde{O}(\log^3(S_\Phi)n^2)^{1/c_1}+n^{2/c_2}}$. Note that from the proof of \Cref{thm:muPRF} the PRF is a $k$-wise independent distribution for $k \geq \log^3 S_\Phi \gg q$, so $\sfE_1$ satisfies the matching condition in its definition. Furthermore since \Cref{thm:muPRF} directly uses \Cref{evalPoly} we know that the top gates in the circuits we use for $\sfE_1$ are all $\mathsf{AND}$ gates.

    For any output $r$ of $\sfE_1$ the restriction from $\sfD_1$ is simply $\rho' = D_{q,N}^{\uparrow}(r)$ (albeit with one bit in each of the $N$ blocks of $r$ resampled after, which in this case means it is fed directly from the seed). The value of $[\rho'(u)=\ast]$ is the same as the $\mathsf{AND}$ of the other $q$ bits in the $u$th block. That is, if $C$ implements the PRF for $\sf{E_1}$, it is $\wedge_{i \in [q]} C(u,i)$.
    
    Now $\rho$ is a composition of $d_\Phi$ restrictions sampled from $\sfD_0$ and $\sfD_1$. The value of $[\rho(u)=\ast]$ is the same as the $\mathsf{AND}$ of $[\rho_i(u)=\ast]$ for each of the constituent restrictions. Since $\sfE_1$ is implemented as a local PRF generated by depth-$d$ size-$S$ circuits, there is a depth-$d$ size-$d_\Phi S$ circuit computing the map $u \mapsto [\rho(u)=\ast]$, which is the same as $[u \in W]$. Since all we have done is added $\mathsf{AND}$ gates on top of a layer of $\mathsf{AND}$ gates there is no increase in depth.
\end{proof}

Combining the above with \Cref{lemma:aw} we get the local PRF we were after.

\tikzstyle{and}=[fill=white, draw=black, shape=circle]
\tikzstyle{generator}=[fill=white, draw=black, shape=rectangle,minimum width=0.8cm,minimum height=0.8cm,inner sep=5pt]

\tikzstyle{new edge style 0}=[-, line width=2.5pt]
\tikzstyle{new edge style 1}=[-, dashed]
\tikzstyle{dotstyle}=[-, dotted]

\begin{figure}
    \centering
\begin{tikzpicture}
	\begin{pgfonlayer}{nodelayer}
		\node [style=none] (13) at (-10.75, 0) {};
		\node [style=none] (14) at (-12, -1) {};
		\node [style=none] (15) at (-10.25, -1) {};
		\node [style=none] (17) at (-11.25, -0.5) {$G^{0}$};
		\node [style=none] (18) at (-12, -1.5) {};
		\node [style=none] (19) at (-11.375, -2) {$seed^{(0)}$};
		\node [style=none] (20) at (-10.75, -1.5) {};
		\node [style=none] (21) at (-10.5, -2) {$u$};
		\node [style=none] (22) at (-10.25, -1.5) {};
		\node [style=none] (23) at (-10.75, -1) {};
		\node [style=none] (24) at (-10.775, 2.25) {};
		\node [style=none] (25) at (-10.75, 0) {};
		\node [style=none] (26) at (-11.5, 0) {};
		\node [style=none] (27) at (-10.75, 0) {};
		\node [style=none] (29) at (-10.975, 0) {};
		\node [style=none] (30) at (-11.5, 2.25) {};
		\node [style=none] (31) at (-11.475, 0) {};
		\node [style=none] (34) at (-12.5, 0.5) {};
		\node [style=none] (36) at (-11, -1) {};
		\node [style=none] (37) at (-6.75, 0) {};
		\node [style=none] (38) at (-8.75, -1) {};
		\node [style=none] (39) at (-6.25, -1) {};
		\node [style=none] (40) at (-7.75, -0.5) {$G^{1}$};
		\node [style=none] (41) at (-8.75, -1.5) {};
		\node [style=none] (43) at (-6.75, -1.5) {};
		\node [style=none] (44) at (-6.5, -2) {$u$};
		\node [style=none] (45) at (-6.25, -1.5) {};
		\node [style=none] (46) at (-6.75, -1) {};
		\node [style=none] (47) at (-6.775, 2.25) {};
		\node [style=none] (48) at (-6.75, 0) {};
		\node [style=none] (49) at (-8.25, 0) {};
		\node [style=none] (50) at (-6.75, 0) {};
		\node [style=none] (52) at (-6.975, 0) {};
		\node [style=none] (53) at (-8.25, 2.25) {};
		\node [style=none] (54) at (-8.225, 0) {};
		\node [style=none] (57) at (-7, -1) {};
		\node [style=none] (58) at (-2.725, 0) {};
		\node [style=none] (59) at (-4.725, -1) {};
		\node [style=none] (60) at (-2.225, -1) {};
		\node [style=none] (61) at (-3.725, -0.5) {$G^1$};
		\node [style=none] (62) at (-4.725, -1.5) {};
		\node [style=none] (64) at (-2.725, -1.5) {};
		\node [style=none] (65) at (-2.475, -2) {$u$};
		\node [style=none] (66) at (-2.225, -1.5) {};
		\node [style=none] (67) at (-2.725, -1) {};
		\node [style=none] (68) at (-2.75, 2.25) {};
		\node [style=none] (69) at (-2.725, 0) {};
		\node [style=none] (70) at (-4.225, 0) {};
		\node [style=none] (71) at (-2.725, 0) {};
		\node [style=none] (73) at (-2.95, 0) {};
		\node [style=none] (74) at (-4.225, 2.25) {};
		\node [style=none] (75) at (-4.2, 0) {};
		\node [style=none] (77) at (-2.975, -1) {};
		\node [style=none] (79) at (-1.975, -0.25) {{ \Large $\dots$}};
		\node [style=none] (86) at (-6.25, 0.5) {$q$};
		\node [style=none] (90) at (-10.25, 0.5) {$6$};
		\node [style=none] (93) at (-7.75, -2) {$seed^{(1)}$};
		\node [style=none] (94) at (-3.725, -2) {$seed^{(2)}$};
		\node [style=none] (103) at (-4.725, 0) {};
		\node [style=none] (104) at (-1.975, 0) {};
		\node [style=none] (105) at (-2, 0.5) {$q$};
		\node [style=none] (109) at (-12.75, 2.1) {};
		\node [style=none] (110) at (-0.475, 2.25) {};
		\node [style=none] (111) at (-1.475, 2.45) {$6 + (d_\Phi-1)q$};
		\node [style=none] (113) at (-5.725, 4) {$[u \in W]$};
		\node [style=none] (114) at (-8.975, 0) {};
		\node [style=none] (115) at (-6.225, 0) {};
		\node [style=none] (116) at (-12.225, 0) {};
		\node [style=none] (117) at (-9.975, 0) {};
		\node [style=and] (119) at (-6.75, 3.25) {$\wedge$};
		\node [style=none] (120) at (-6.75, 4.25) {};
	\end{pgfonlayer}
	\begin{pgfonlayer}{edgelayer}
		\draw (13.center) to (15.center);
		\draw (14.center) to (15.center);
		\draw (14.center) to (18.center);
		\draw (15.center) to (22.center);
		\draw (23.center) to (20.center);
		\draw (14.center) to (26.center);
		\draw (26.center) to (25.center);
		\draw (24.center) to (27.center);
		\draw (30.center) to (31.center);
		\draw (37.center) to (39.center);
		\draw (38.center) to (39.center);
		\draw (38.center) to (41.center);
		\draw (39.center) to (45.center);
		\draw (46.center) to (43.center);
		\draw (38.center) to (49.center);
		\draw (49.center) to (48.center);
		\draw (47.center) to (50.center);
		\draw (53.center) to (54.center);
		\draw (58.center) to (60.center);
		\draw (59.center) to (60.center);
		\draw (59.center) to (62.center);
		\draw (60.center) to (66.center);
		\draw (67.center) to (64.center);
		\draw (59.center) to (70.center);
		\draw (70.center) to (69.center);
		\draw (68.center) to (71.center);
		\draw (74.center) to (75.center);
		\draw [style=new edge style 1, bend left, looseness=1.25] (103.center) to (104.center);
		\draw [style=new edge style 1, in=-165, out=-15, looseness=0.75] (109.center) to (110.center);
		\draw [style=new edge style 1, bend left, looseness=1.25] (114.center) to (115.center);
		\draw [style=new edge style 1, bend left, looseness=1.25] (116.center) to (117.center);
		\draw (30.center) to (119);
		\draw (24.center) to (119);
		\draw (53.center) to (119);
		\draw (47.center) to (119);
		\draw (74.center) to (119);
		\draw (68.center) to (119);
		\draw (119) to (120.center);
	\end{pgfonlayer}
\end{tikzpicture}

\caption{The construction of a circuit that on input $u \in \Q^n \equiv [N]$ computes $[u \in W]$ (see \Cref{lemma:constructingW}). Here the circuit $G^0$, on input $u$ and a seed, computes the $u$th part of the corresponding string $r \in \Q^{(6+1)N}$ in the support of $\sfE_0$. Similarly the circuit $G^1$, on input $u$ and a seed, computes the $u$th part of the corresponding string $r \in \Q^{(q+1)N}$ in the support of $\sfE_1$. We ensure that both circuits $G^0$ and $G^1$ have a top layer of $\mathsf{AND}$ gates, so there is no increase in depth in this figure.}\label{fig:UinW}
\end{figure}
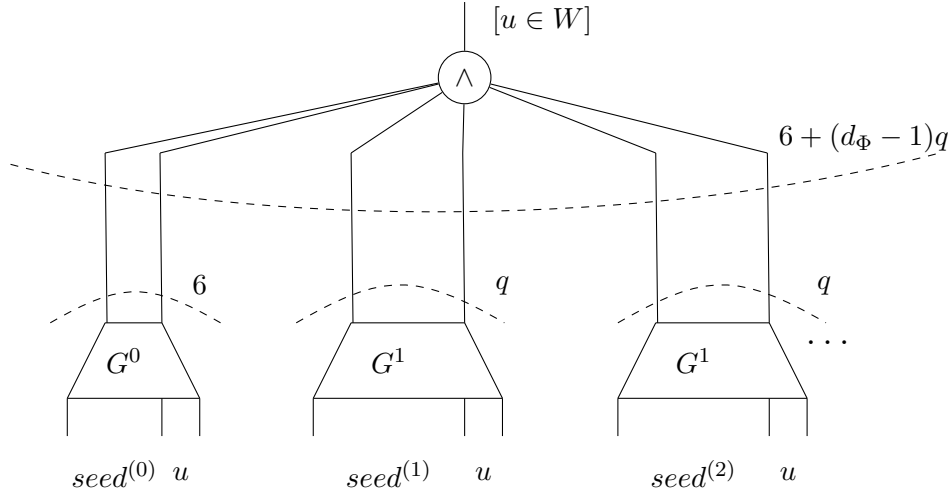

\begin{corollary}\label{cor:pre-unconditional}
    Let $c_1,c_2 \in \mathbb{N}$. There is a $(n,d,S,d_\Phi,S_\Phi)$-PRF where $d=c_1+c_2+4$ and $S=2^{O(d_\Phi)}(\log S_\Phi)^{d_\Phi}2^{\tilde{O}(\log^3(S_\Phi)n^2)^{1/c_1}+n^{2/c_2}}$. From \Cref{lem:prfcorollary} this also means that there is no $(n,d,S,d_\Phi,S_\Phi)$-natural proof.
\end{corollary}

Finally we analyze the implications for natural proofs. We see what parameters $d_\Phi,S_\Phi$ are needed to get a $2^{n^\eps}$ lower bound for depth $d$ circuits.

\begin{corollary}\label{thm:unconditional-natural-proofs}
    Let $d \in \mathbb{N}, d \geq 7$. There is no constant depth (or even a $\log \log$-depth) and polynomial size $\AC_0$ property that can yield a natural proof lower bound larger than $2^{n^{7/(d-5)}}$ for circuits of depth $d$.
    
    More generally, let $\eps \in (0,1), \eps \geq 2/(d-6)$. Let $\Phi = \{\Phi_n : \Q^{2^n} \rightarrow \Q\}$ be a family of $(n, d, 2^{n^\eps}, d_{\Phi}, S_{\Phi}, 8/S_\Phi)$-natural properties. If $S_\Phi \leq 2^{n^\alpha}$, then either $d_\Phi \geq n^\eps/(\alpha \log n)$ or $\eps \leq (3\alpha+4)/(d-5)$.
\end{corollary}

\begin{proof}
    We start by proving the more general case. Set $c_2 = \lceil2/\eps\rceil \leq 2/\eps+1$. Now $c_1 = d-c_2-4\geq d-2/\eps-5$. We know such a choice is possible because $\eps \geq 2/(d-6)$.
    
    If $S_\Phi \leq 2^{n^\alpha}$, then by \Cref{cor:pre-unconditional} the maximum lower bound provable is at most $2^{O(d_\Phi)}n^{\alpha d_\Phi} \cdot 2^{\tilde{O}(n^{3\alpha+2})^{1/c_1}+n^{2/c_2}}$. For $2^{n^\eps}$ to be at most this bound we need $\alpha d_\Phi \log n \geq n^\eps$, $2/c_2 \geq \eps$ or $(3\alpha+2)/c_1 \geq \eps$. We have ensured that $2/c_2 < \eps$. Simplifying the other two conditions we get that either $d_\Phi \geq n^\eps/(\alpha \log n)$ or $\eps \leq (3\alpha+4)/(d-5)$.

    For the last part a polynomial size property has $S_\Phi = 2^{O(n)}$ and so we can substitute $\alpha=1$.
\end{proof}

\section{A candidate tight natural-proofs barrier}\label{sec:candidate}

Here we will propose a concrete PRF, which is computable by depth-$d$ circuits of size $2^{\tilde O(n^{\frac{1}{d-1}})}$, and conjecture that it is secure against $\AC_0$ distinguishers. We also discuss the possibility that variants of it may be secure against all efficient distinguishers (not just $\AC_0$ distinguishers). If this is indeed the case, then it follows that the state of the art in constant-depth circuit lower-bounds cannot be improved by \textit{any} natural proof. 

The construction is very simple, it is based on composition of randomly-chosen functions:

\begin{definition}
    Let $\Bal_m$ denote the set of all \textit{balanced} Boolean functions $H:\ZO^m\to\ZO$, meaning functions such that $|H^{-1}(1)| = |H^{-1}(0)|$. Let $H_{m}$ denote a random function sampled uniformly from $\mathsf{Bal}_m$. Let $H_{n,m}$ denote a random Boolean function over $\ZO^n$ sampled according to the following procedure:
    \begin{enumerate}
        \item Sample some set $S \subseteq [n]$ of size $m$, uniformly at random.
        \item Sample a random Boolean function $H_m:\ZO^m\to\ZO$ uniformly from $\Bal_m$.
        \item Then set $H_{n,m}(x) = H_m(x|_S)$ for all $x$.
    \end{enumerate}
    Finally, let $F_{d,n,m}$ denote a random function $F_d:\ZO^n\to\ZO$, whose distribution we define inductively as follows. 
    \begin{itemize}
        \item For $d = 2$, we sample $F_{2,n,m}(x) = H_{n,m}(x)$. 
        
        \item Then for $d \ge 3$ we define $F_{d,n,m}$ by block composition: $$F_{d,n,m}(x)= H_m(F_{d-1,n,m}(x), \dots, F_{d-1,n,m}(x)),$$ i.e. we sample $H_m$ and feed into it $m$ independently-sampled copies of $F_{d-1,n,m}$.
    \end{itemize}

    We then use $F_{d,n}$ to denote a random function $F_{d,n,m}$ where $m = \tilde O(n^{\frac{1}{d-1}})$. If one wishes to be more concrete, one may use $m = (\log n)^2\cdot n^{\frac{1}{d-1}}$.
\end{definition}

Notice that $F_{d,n}$ is a Boolean function on $\tilde O(n)$ bits. By composing the DNF/CNF representations of each individual $H_m$ and $H_{n,m}$, we find that $F_{d,n}$ is computable by a depth-$d$ circuit of size $S = 2^{\tilde O(n^{\frac{1}{d-1}})}$. We conjecture it is secure:

\begin{conjecture}
    The distribution of $F_{d,n}$ is a function generator that fools all $\AC_0$ distinguishers. 
\end{conjecture}

It follows as a consequence of this conjecture that $\AC_0$ natural proofs, such as proofs based on the Switching Lemma, cannot show a $2^{\tilde\omega(n^{\frac{1}{d-1}})}$-size lower-bound against depth-$d$ circuits. Security against $\AC_0$-distinguishers is something which one might hope to prove unconditionally, and so we consider this conjecture to be an interesting open problem.

One may also wonder whether such superexponential lower bounds can be proven even using properties outside of $\AC_0$, say computable in polytime. The above class of distributions, at least in the case $d=3$, can actually de distinguished from random by polytime tests. However if one were to modify the definition of $F_{3,n}(x)$ to $H_{m}(G_1(x),\dots,G_m(x))$ with half of the $G_i$s being uniformly random functions of a random $m$ of the $n$ inputs, and the other half being parities of a random $m$ of the $n$ input bits, then we are unable to find polytime tests that distinguish it from random. We find it interesting to know to what extent we can show that polytime tests distinguish depth-3 size-$2^{O(\sqrt{n})}$ distributions from uniform. Not being able to distinguish them would be akin to saying that there are PRGs in this bounded depth circuit class (and hence a natural-proofs barrier), and so we consider pushing this boundary to also be a worthwhile task.

\section{Conclusion}

The main problem which we leave open is an exact answer to \Cref{qu:prf}. Concretely, prove or disprove whether there are PRFs computable by depth-3 circuits of size $2^{O(\sqrt{n})}$ fooling $\AC_0$ distinguishers. Furthermore, we believe it is worthwhile to consider other models of computation and try to construct unconditional efficient PRFs computable in size close to the best known lower bounds. In \Cref{sec:all-natural} we also looked at some constant depth lower bound proofs that we couldn't find natural variants for. We consider finding such variants to be another interesting open problem.

\paragraph*{Acknowledgements}
We thank Luís Pereira for pointing out a proof in the literature that we could not naturalize. We also thank anonymous reviewers whose comments helped us improve the presentation.

\bibliographystyle{plain}
\bibliography{bibliography}

\appendix

\newcommand{\AND}{\mathsf{AND}}
\newcommand{\OR}{\mathsf{OR}}
\newcommand{\NOT}{\mathsf{NOT}}

\section{Deferred Proofs}\label{sec:deferredproofs}

\subsection{From the preliminaries}

We stated this folklore proposition about PRFs.

\fooling*

\begin{proof}
    The first item is true since if an algorithm in $\mathcal{A}$ can distinguish a marginal of $\mu_1$ from uniform, it can be used as-is to distinguish $\mu_1$ from uniform.

    The second item is a classic use case of the hybrid argument. Note that for any algorithm $A \in \mathcal{A}$, we have that $A$ cannot distinguish $\mu_1 \times \mu_2$ from $\mu_1 \times \Unif_{N_2}$ any better than it can distinguish $\mu_2$ from $\Unif_{N_2}$:
    \begin{align*}
        \left| \E_{X \sim \mu_1} \E_{Y \sim \mu_2} [A(X,Y)] - \E_{X \sim \mu_1} \E_{Y \sim \Unif_{N_2}} [A(X,Y)] \right| &= \left| \E_{X \sim \mu_1} \left[\E_{Y \sim \mu_2} [A(X,Y)] - \E_{Y \sim \Unif_{N_2}} [A(X,Y)]\right] \right|\\
        &\leq \max_{X \in \mathsf{supp}(\mu_1)} \left| \E_{Y \sim \mu_2} [A(X,Y)] - \E_{Y \sim \Unif_{N_2}} [A(X,Y)] \right| < \epsilon_2.
    \end{align*}
    The last inequality follows since we can hardcode the $X$ in the algorithm $A$ and we will have an algorithm $A' \in \mathcal{A}$ tasked with distinguishing $\mu_2$ from $\Unif_{N_2}$. The same argument can be applied to show that $A$ cannot distinguish $\mu_1 \times \Unif_{N_2}$ from $\Unif_{N_1} \times \Unif_{N_2}$ any better than it can distinguish $\mu_1$ from $\Unif_{N_1}$. And hence by the triangle inequality, it is $(\epsilon_1+\epsilon_2)$-fooled by $\mu_1 \times \mu_2$.
\end{proof}

We also needed this observation about multiplicative approximate counting.

\approxcount*

\begin{proof}
    From the well-known fact proven by Ajtai and Ben-Or \cite{ajtai1984theorem} we know that for any $t \in [n]$ there is a function $\mathsf{GapCount}^+_{n,t,\eps} : \ZO^n \to \ZO$ satisfying
    \[
    \mathsf{GapCount}^+_{n,t,\eps}(x) = \begin{cases}
       0 & \text{if } |x| \le t\\
       1 & \text{if } |x| \ge t + \eps \cdot n
    \end{cases}
    \]
    that can, for any $d \ge 2$, be computed by a circuit of depth $d+1$ and size $O\left(n^2 \cdot 2^{\tilde{O}\left(\left(\frac{\log n}{\eps^2}\right)^{1/(d-1)}\right)}\right)$. We include a proof of this as \Cref{clm:additive-approx-count} for completeness.

    Using this we can show that for any threshold $t$ and approximation $\eps \in [0,1]$ we can build a slightly larger circuit computing $\mathsf{GapCount}^{\times}_{n,t,\eps}$. We start by padding the input to length $n'$ with 0s, so that $t/n' < \eps/10$. Now consider an OR gate $v$ over $\lceil\frac{n'}{t}\rceil$ coordinates from $[n']$, chosen at random with replacement.
    Then:
    \begin{align*}
        |x| \le t & \implies \Pr[v(x) = 0] = \left(1 - \frac{|x|}{n'}\right)^{\frac{n'}{t}} \ge \left(1 - \frac t {n'}\right)^{\frac{n'}{t}} \ge \frac{1}{e} - \frac{t}{en'}\\
        |x| \ge (1 + \eps) \cdot t & \implies \Pr[v(x) = 0] = \left(1 - \frac{|x|}{n'}\right)^{\frac{n'}{t}} \le \left(1 - (1+\eps) \cdot\frac t {n'}\right)^{\frac{n'}{t}} \le \frac{1}{e^{1+\eps}}
    \end{align*}
    Now suppose that we independently sample $m = O(n/\eps^2)$ such ORs. By Chernoff and a union bound, there is a setting of the randomness such that every $x$ with $|x| \le t$ will have at least $\left(e^{-1} - t/en' - \eps/10 \right)m$ such ORs set to $0$, and every $x$ with $|x| \ge (1 + \eps) \cdot t$ will have at most $\left(e^{-(1+\eps)} + \eps/10 \right)m$ many. We have $e^{-1} - e^{-(1+\eps)} \geq \frac{\eps}{e} - o(\eps)$, so the two cases can now be distinguished using $\mathsf{GapCount}^+_{m,(e^{-(1+\eps)}+\eps/10)m,\eps/10}$.

    Even with $\eps = 1/\log^r n$, there is a circuit of depth $2r+4$ and size $\poly(n)$ that computes $\mathsf{GapCount}^+_{m,(e^{-(1+\eps)}+\eps/10)m,\eps/10}$. By feeding this circuit the $m$ OR gates (thus increasing the depth by $1$ and the size by $2m$), we have a circuit computing $\mathsf{GapCount}^{\times}_{n,t,\eps}(x)$ as well.
\end{proof}

\begin{claim}\label{clm:additive-approx-count}
    There is a function $\mathsf{GapCount}^+_{n,t,\eps} : \ZO^n \to \ZO$ satisfying
    \[
    \mathsf{GapCount}^+_{n,t,\eps}(x) = \begin{cases}
       0 & \text{if } |x| \le t\\
       1 & \text{if } |x| \ge t + \eps \cdot n
    \end{cases}
    \]
    that can, for any $d \ge 2$, be computed by a circuit of depth $d+1$ and size $O\left(n^2 \cdot 2^{\tilde{O}\left(\left(\frac{\log n}{\eps^2}\right)^{1/(d-1)}\right)}\right)$.   
\end{claim}

The proof provided below follows a manuscript by Oded Goldreich~\cite{odedapproxmaj}. Since the statement above is a bit more general we are reproducing the (essentially same) proof.

\begin{proof}
    For a string $s \in \{0,1\}^m$ let $\mathsf{wt}(s)$ denote the fractional Hamming weight $\frac{|\{i \in [m] \mid s_i=1\}|}{m}$.
    
    The circuit is conceptually simple and will just be the $\OR_n$ of $\AND_n$s of threshold functions on $O((\log n)/\eps^2)$ bits (or alternatively an $\AND_n$ of $\OR_n$s of such functions). Since a threshold function on $m$ bits can be computed in depth $d$ and size $2^{\tilde{O}\left(m^{1/(d-1)}\right)}$ (\cite{hastad-thesis}, or see footnote\footnote{Consider the function $f_k : (\ZO^{\lceil \log m \rceil})^{k} \to \ZO^{\lceil \log m \rceil}$, where the inputs are interpreted as $a_1,\dots,a_k \in [m]$, that computes $\sum_{i \in [k]} a_i$ (with the promise that the sum is always at most $m$). Since there are only $k\lceil\log m\rceil$ input bits, each output bit of the function can be written as an $\OR$ of $\AND$s of literals (or an $\AND$ of $\OR$s of literals) of size $2^{k \lceil \log m \rceil}$. By composing together $d-1$ layers of $f_k$ with $k = m^{1/(d-1)}$ we get a function that computes the exact count of $1$s present in $m$ input bits. Replace the top $f_k$ with a similarly sized function that outputs $1$ iff $\sum a_i > c$ to make the composition a threshold function. Choosing circuits for each instance of $f_k$ to make $\OR$s and $\AND$s collapse, we obtain a circuit of depth $d$ and size $O(m \cdot \log m \cdot 2^{k \lceil \log m \rceil})$.}), and by collapsing together the second and third layers, this is a depth $d+1$ circuit of the requisite size. We will now see why such a circuit can exhibit the requisite behaviour.

    There is a choice of subsets $S_1,\dots,S_{n^2} \subseteq [n]$, each of size $O((\log n)/\eps^2)$, such that for every $x \in \ZO^n$ more than $n^2-n$ of the sets $S$ are ``good'' in that they satisfy $|\mathsf{wt}(x)-\mathsf{wt}(x_{S})| \leq \eps/3$. To see this note that if we fix any $x$ and choose each $S_i$ by choosing $(15\ln n)/\epsilon^2$ indices at random with replacement, by a Chernoff bound each $S_i$ has a $< n^{-3}$ probability of being a bad set. The probability that at least $n$ are bad is then at most $\binom{n^2}{n} (n^{-3})^n \leq n^{-n}$. By a union bound over all $x$ there is a setting of $S_1,\dots,S_{n^2}$ satisfying the property.

    Hence by computing whether $\mathsf{wt}(x_S) \geq t/n + \eps/2$ for each set, we get $n^2$ bits $b_1,\dots,b_{n^2}$ such that the number of $1$s is larger than $n^2-n$ if $|x| \ge t + \eps \cdot n$, and smaller than $n$ if $|x| \leq t$. This is easy to differentiate as follows. Label the bits as $\{b_{i,j}\}_{i,j \in [n]}$. The circuit $\vee_{i \in [n]} \wedge_{j \in [n]} b_{i,j}$ distinguishes the two cases: If there are less than $n$ $0$s it must be the case that some $\AND$ outputs $1$ so the circuit must output $1$, and if there are less than $n$ $1$s every $\AND$ must output $0$ so the circuit must output $0$. (A similar argument works using $\wedge_{i \in [n]} \vee_{j \in [n]} b_{i,j}$ instead.) This completes the proof.
\end{proof}

\subsection{On evaluating polynomials in constant depth}

\altevalPoly*
\begin{proof}
Let $n \in \mathbb{N}$. We work in the field $\mathbb{F} = \mathbb{F}_{2^n}$. This is the same as the field $\mathbb{F}_2[x] / p(x)$ for an irreducible $p(x)$ of degree $n$. The elements of $\mathbb{F}$ are thus identified with $\ZO^n$ as just the string of coefficients.

Given as input $\alpha^{(0)} ,\dots ,\alpha^{(k - 1)} , y \in \mathbb{F}$, we want to compute $\sum_{i = 0}^{k - 1} \alpha^{(i)} y^i$.

\begin{enumerate}
    \item Healy et al.~\cite{HealyViola} shows that computing $y^{2^i}$ is quite easy: $\binom{2^i}{j}$ is even for every $j \neq 0 \text{ or } 2^i$, so if $y = \sum y_h x^h$, then $y^{2^i} = \sum_h y_h x^{h 2^i}$. We can store, for each $h \in [n]$ and $i \in [\lfloor \log k\rfloor]$, the elements $x^{h 2^i}$. Hence $y^{2^i}$ is the $n$-bit string whose $r$th bit is $\oplus_{\{h : (x^{h 2^i})_r = 1\}} y_h$. Hence the circuit computing $y \mapsto (y^1 , y^2 , y^4 ,\ldots , y^{2^{\lfloor \log k\rfloor}})$ is implemented in parallel by $n \lfloor \log k\rfloor$ parities of arity at most $n$ each.
    \item Now we want to compute an element of the form $\alpha^{(i)} y^i$. This is the multiplication of $t \leq 1 +\lfloor \log k\rfloor$ terms: $\alpha^{(i)}$ and the appropriate $y^{2^j}$s. Following~\cite{HealyViola} we compute this with a circuit for iterated multiplication of $t$ elements $\beta^{(1)} ,\ldots ,\beta^{(t)} \in \mathbb{F}_2[x]$, only afterwards reducing it to $\mathbb{F}$. The iterated multiplication is done via the Discrete Fourier Transform. Let $\widehat{\mathbb{F}} = \mathbb{F}_{2^m}$ with $m = \lceil \log(n t)\rceil$, and let $S \subset \widehat{\mathbb{F}}$ be a set of non-zero elements of size $(n - 1)t + 1$. It is well known that the polynomial evaluation map for degree $(n - 1)t$ polynomials, $M : (p_h)_{h \in [0 ,\ldots ,(n - 1)t]} \mapsto (p(g))_{g \in S}$ is an invertible $\widehat{\mathbb{F}}$-linear map. We use this to get $(\beta^{(i)}(g))_{g \in S}$ for each $i$. From these we define $\beta^{\times}(g) = \prod_i \beta^{(i)}(g) \in \widehat{\mathbb{F}}$. These are the evaluations of the degree $(n - 1)t$ polynomial $\beta^{\times} = \prod_i \beta^{(i)}$. We then use $M^{- 1}$ to recover the $(n - 1) t + 1$ coefficients of $\beta^{\times}$. Finally we need to reduce this to a polynomial in $\mathbb{F}$. We explain how to do each step below.

    \begin{enumerate}
        \item Since $M$ is $\widehat{\mathbb{F}}$-linear and the coefficients of each $\beta^{(i)}$ are in $\mathbb{F}_2$, $\beta^{(i)}(g)$ is the $m$-bit string whose $r$th bit is $\oplus_{h : (M_{g , h})_r = 1} \beta_h^{(i)}$. So from $(\beta^{(i)})_{i \in [t]}$ we compute $(\beta^{(i)}(g))_{i \in [t], g \in S}$ by parallelly using at most $n t^2 m$ parities of arity at most $n$. 
        \item Each bit of $\beta^{\times}(g)$ is a function of $(\beta^{(i)}(g))_{i \in [t]}$, so each bit is a function of at most $m t$ bits that have already been computed. To simplify the next step we will instead compute terms of the form $(M^{- 1})_{h , g} \cdot \beta^{\times}(g)$, which also depend on the same $m t$ bits. Hence each bit can be computed by both a CNF and a DNF of size $2^{m t}$. There are $m n^2 t^2$ bits of output in $((M^{- 1})_{h , g} \cdot \beta^{\times}(g))_{h \in [0 ,\ldots ,(n - 1)t], g \in S}$.
        \item We can now recover the coefficients of $\beta^{\times}$ from $(\beta^{\times}(g))_{g \in S}$ as $\beta_h^{\times} = \sum_g (M^{- 1})_{h , g} \cdot \beta^{\times}(g)$. Since each $\beta^{(i)}$ had coefficients in $\mathbb{F}_2$, so will $\beta^{\times}$ and so we only care about one bit of each coefficient. The relevant bit for the $h$th coefficient is simply $\oplus_g (M_{h , g}^{- 1} \cdot \beta^{\times}(g))_0$, so it is a parity of arity at most $n t$, and there are at most $n t$ such coefficients to compute.
        \item Finally we need to simplify the computed degree-$(n - 1)t$ polynomial $\beta^{\times} \in \mathbb{F}_2[x]$ to its modulus in $\mathbb{F}$. For this we simply pre-store the representations of $x^r \in \mathbb{F}$ for each $r \in [n ,\ldots ,(n - 1)t]$. Our final polynomial's $h$th coefficient is then $\beta_h^{\times} + \sum_{r \geq n} (x^r)_h$. This is a parity of arity at most $n t$, and there are at most $n$ of them.
    \end{enumerate}

    \item Using $k$ copies of the circuit above we construct each $\alpha^{(i)} y^i \in \mathbb{F}$. We now sum up over all $k$ terms to get $\sum_{i = 0}^{k - 1} \alpha^{(i)} y^i$. This is achieved by $n$ parallel parities of arity $k$, and that is our output.
\end{enumerate}

Looking at the above construction we get a circuit with the following layers with $t = 1 +\lfloor \log k\rfloor$ and $m =\lceil \log(n t)\rceil$: 

\begin{itemize}
    \item $n$ $\oplus_k$ gates,
    \item $k n$ $\oplus_{n t}$ gates,
    \item $k n t$ $\oplus_{n t}$ gates,
    \item $k n^2 t^2 m$ CNFs of size $2^{m t}$,
    \item $k n t^2 m$ $\oplus_n$ gates,
    \item $k n \lfloor \log k\rfloor$ $\oplus_n$ gates.
\end{itemize}
This simplifies to a top layer of $n$ $\oplus_{k n^2 t^2}$ gates, a middle layer of $k n^2 t^2 m$ CNFs of size $2^{m t}$, and a bottom layer of $k n t$ $\oplus_{n^2}$ gates.

Using unbounded arity AND/OR gates we can implement this in constant depth as follows. For any $c \in \mathbb{N}$ the parity of $s$ bits can be computed in depth $c+1$ and size $s2^{s^{1/c}}$. Such circuits exist with an $\AND$ gate at the top and they also exist with an $\OR$ gate at the top. We implement the top layer of parities in depth $c_1+ 1$ and size $\poly(k,n) 2^{(k n^2 t^2)^{1/c_1}}$ with $\AND$ gates at the top. For each input literal to the parity circuit we can attach a CNF/DNF of the middle layer appropriately so that the top gate of the CNF/DNF merges with the gate that the literal fed in to. The circuit now has depth $c_1+2$ and size $2^{m t} \cdot n 2^{(k n^2 t^2)^{1/c_1}}$. Finally for every input literal of this circuit we attach a depth $c_2+1$ size $n^2 2^{n^{2/c_2}}$ parity circuit computing said input, choosing the parity circuit appropriately so that the top gate of the parity merges with the gate that the literal fed in to. Since $2^{m t} \leq 2^{\tilde{O}(nk)}$, for any fixed constants $c_1,c_2 \in \mathbb{N}$ the final circuit is of depth $c_1 + c_2 + 2$ and size at most \[
\poly(k,n) 2^{m t} \cdot \left(n \cdot 2^{(kn^2t^2)^{1/c_1}} + k n t \cdot 2^{n^{2/c_2}}\right) \leq 2^{\widetilde{O}(kn^2)^{1/c_1} + n^{2/c_2}}.
\]
\end{proof}

\newcommand{\AND}{\mathsf{AND}}
\newcommand{\OR}{\mathsf{OR}}
\newcommand{\NOT}{\mathsf{NOT}}

\section{Properties used in some Top-Down Lower Bounds}\label{sec:topdownproperties}

The depth 3 lower bounds of Håstad, Jukna and Pudlák follow by reasoning that any small depth 3 circuit computing a certain function $f$ must contain a width-$k$ term/clause that:
\begin{itemize}
    \item rejects every input in a set $T \subseteq \ZO^n$, and
    \item accepts an input $y$ such that for every $Q \subseteq [n]$ of size at most $k$, there is an $x \in T$ such that $x_Q = y_Q$.
\end{itemize}

Such a $y$ is called a \textbf{$k$-limit} of the set $T$. Note that no width-$k$ term/clause can reject every input in $T$ while accepting a $k$-limit of $T$. Hence the small depth 3 circuit could not have been computing $f$. In this section we go over the properties of $f$ that Håstad et al. use in order to carry out the above proof (without going over how it is carried out). These properties are not natural. We will then see a natural property with which the above proof can also be carried out, albeit yielding slightly weaker lower bounds.

\newcommand{\HJP}{\mathrm{HJP}}

\subsection*{Håstad et al., depth 3 with bounded bottom fan-in}

Håstad et al. use the following Majority-like property to show a lower bound of $\ell$ against $\Pi_3^k$ circuits (depth 3 $\wedge\vee\wedge$ circuits with the bottom gates having fan-in $k$). $\Phi_{\HJP,k,\ell}(f)=1$ if $\exists s \in [n], S \subseteq \binom{n}{s}$ such that:
\begin{itemize}
    \item $f(x)=1$ for all $x \in S$,
    \item $f(x)=0$ for all $x<y \in S$, and
    \item $|S| \geq \ell \cdot k^s$.
\end{itemize}

As an example, this property shows that the Threshold$_{n,n/k}$ function requires $e^{(n/k) \cdot (1-o(1))}$ size $\Pi_3^k$ circuits. A lower bound of $e^{(n/2k) \cdot (1-o(1))}$ also follows for Majority since it embeds the above functions on $n/2$ bits.

(\cite{GurumukhaniPPST24} make a very similar statement in the special case when $k=3$ and with $s$ fixed to $n/2$: if $f(x)=1$ for all $x \in S$ and $f(x)=0$ \emph{for all $x$ with $|x|<n/2$}, then $|S| \geq \ell \cdot 2.554^{n/2}$ suffices to get the lower bound of $\ell$. While the resulting property looks similar to that used by Håstad et al., their proof is significantly different and does not even use $k$-limits.)

\subsection*{Håstad et al., depth 3}

To prove their depth $3$ lower bound, Håstad et al. uses restrictions. Every circuit has a restriction of some variables to $1$s that kills all bottom gates with many negated inputs. (Bounding the ``negated bottom fan-in'' is sufficient for the previous lower bound.) Starting with a circuit of size $2^{\sqrt n}$, they can set around $n/2 - \sqrt{n}$ bits to $1$ and make the ``negated fan-in'' smaller than $k \approx\sqrt n$. The restricted function is a threshold function on $\approx n/2$ bits checking if the Hamming weight is at least $\approx \sqrt n$. We know $\Pi_3^{k}$ circuits for this require size $2^{\Omega(\sqrt n)}$. Using this method they get a specific lower bound of $\approx 2^{0.849\sqrt{n}}$ for the size of $\Pi_3$ circuits computing Majority. Since Majority is self-dual, it is also a lower bound on the size of $\Sigma_3$ circuits.

So the property that they use for the $\Pi_3$ lower bound is something like the following. For all restrictions $\rho$ setting all but $m := n/2 + \sqrt n$ of the $n$ input bits to $1$, the restricted function $f|_{\rho}:\ZO^m \to \ZO$ has the property $\Phi_{\HJP,\sqrt n,\mathrm{exp}(\sqrt n)}$.

None of the above properties are large, although they give rise to constructible properties by using approximate counting. Håstad et al. also use a similar property adapted to work for Parity, but we will skip that due to its similarity.

\subsection*{Natural property, depth 3 with bounded bottom fan-in}

The natural property (cite relevant places) is quite simple. For instance we can consider $\Phi(f) = 1$ if and only if $\exists S \subseteq \ZO^n$ such that:
\begin{itemize}
    \item $f(x)=0$ for all $x \in S$,
    \item the sensitivity of $f$ at $x$ is at least $(1-100\log n/\sqrt{n}) \cdot n/2$ for all $x \in S$, and
    \item $|S| \geq 2^n/n^2$.
\end{itemize}

We will now see that this property gives rise to a constructible property that is large and is useful: any function with the property requires $2^{(1-o(1)) \cdot n/2(k+1)}$-size $\Pi_3^k$ circuits. To create the constructible property we can compute the sensitivity of each of the $2^n$ points and use approximate counting on the $2^n$ inputs to reject when the number of $1$-inputs with high sensitivity is smaller than $2^n/n^2$ and accept when it is larger than $2 \cdot 2^n/n^2$. The largeness of this constructible property is easy to verify since there is a distance-3 code of size $2^n/2(n+1)$, and for a random function each input that is a codeword has independent function values and sensitivities.

\begin{claim}
    Let $f : \ZO^n \to \ZO$ be such that $\Phi(f)=1$. Any $\Pi_3^k$ circuit computing $f$ must have size at least $2^{(1-o(1)) \cdot n/2(k+1)}$.
\end{claim}

\begin{proof}
    Let $C$ be a size-$2^r$ $\Pi_3^k$ circuit computing $f$ with $r = (1-1/\log n) \cdot n/2(k+1)$ (the $o(1)$ term can be optimized, we use this as a simple example). We will show that this eventually leads to a contradiction. We know $C$ must reject all of $S$, and it is an $\AND$ of at most $2^r$ $\Sigma_2^k$ circuits. So there must be a $\Sigma_2^k$ circuit, let's call it $C'$, that rejects a set $S' \subseteq S$ of size at least $|S|/2^r$ while accepting all of $f^{-1}(1)$.
    
    At this point we invoke a lemma from Smal and Talebanfard~\cite{SmalT2018prediction}. Solely based on the fact that $|S'| \geq 2^{n-2\log n-r}$ the lemma states that $\Pr_{x \sim S', i \sim [n]}[x \oplus e_i\text{ is a }k\text{-limit of }S'] \geq 1-(2\log n + r)(k+1)/n$. Our choice of $r$ was to ensure that this probability is a bit larger than $1/2$, so that the following set $A$ is sizeable.
    
    Let $A = \{(x,i) \in S' \times [n] \mid f(x \oplus e_i) = 1 \wedge x \oplus e_i\text{ is a }k\text{-limit of }S'\}$. $A$ has size at least $|S'| \cdot (1-100\log n/\sqrt{n}) \cdot n/2 - |S'| \cdot (1-1/\log n) \cdot n/2 \geq |S'|$.
    
    Let $T = \{y \in \ZO^n \mid \exists (x,i) \in A \text{ s.t. } y = x \oplus e_i\}$. Since each $y \in T$ can come from at most $n$ elements of $A$, $|T| \geq |S'|/n$. Since $T \subseteq f^{-1}(1)$, we know $C'$ accepts all of $T$. Since $C'$ is the $\OR$ of at most $2^r$ terms of width $k$ there must be a term $C''$ of width $k$ that accepts a set $T' \subseteq T$ of size at least $|T|/2^r$ (and hence non-empty) while rejecting all of $S'$.
    
    Since every $y \in T'$ is a $k$-limit of $S'$, and $y$ must be accepted while all of $S'$ is rejected, the proof is concluded by contradiction.
\end{proof}

\subsection*{Natural property, depth 3}

Surprisingly the same property $\Phi$ is also useful against general depth 3 circuits: any function with the property also requires $2^{(1-o(1)) \cdot \sqrt{n}/2}$-size $\Pi_3$ circuits. Remarkably the proof of usefulness is almost the same as in the previous case. There is no use of restrictions that bounds the bottom fan-in, instead a large set of accepted inputs manages the same.

\begin{claim}
    Let $f : \ZO^n \to \ZO$ be such that $\Phi(f)=1$. Any $\Pi_3$ circuit computing $f$ must have size at least $2^{(1-o(1))\cdot\sqrt{n}/2}$.
\end{claim}

\begin{proof}
    We follow the previous proof with $r = (1-1/\log n)\sqrt{n}/2$ and using $k=\sqrt{n}$ in the definition of $A$. We still get $|A| \geq |S'|$ and we reach a term $C''$ (of unbounded width) that accepts a set $T'$ of size at least $|T|/2^r \geq |S|/2^{2r+\log n} \geq 2^{n-2r-3\log n}$. However, for a term to accept at least these many inputs its width must be at most $2r + 3\log n$. Letting $Q$ be the indices read by $C''$ (at most $2r + 3\log n \leq \sqrt{n}$ indices) the proof then continues as before since every $y \in T'$ is a $k$-limit of $S'$.
\end{proof}

We can modify the property to also make the analogous claim for $1$-inputs so that the lower bound is against any depth 3 circuit, and the property still remains natural.

\paragraph{A holdout:} While these vaguely similar proofs using natural properties manage to get slightly weaker lower bounds that hold for both Majority and Parity, we should note that this is not always the case. Håstad et al. also study a function $\AND_{\sqrt{n}} \circ \OR_{\sqrt{n}} \circ \AND_2$. The lower bound they get here is much weaker than the ones they got for Majority and Parity, but it shows an exponential separation between the $\Pi_3^k$ and $\Sigma_3^k$ complexities. We have been unable to find a natural property that shows a similar separation for this function.

\end{document}